\documentclass[10pt]{article}

\usepackage[a4paper, margin=1.2in]{geometry}
\usepackage[dvipdfm]{graphicx}

\usepackage{amsmath, amssymb, amsfonts, mathtools, bm, mathrsfs, ulem, amsthm} 
\usepackage{algorithm, algorithmic}

\usepackage{caption}
\usepackage{titlesec}

\usepackage{ascmac, enumitem}
\usepackage{csquotes, indentfirst}
\usepackage[affil-it]{authblk}
\usepackage{float, xcolor}

\usepackage{hyperref}
\hypersetup{
    colorlinks=true,
    citecolor=blue,
    linkcolor=red,
    urlcolor=orange,
}

\renewcommand{\appendixname}{Appendix}
\newcommand{\startappendix}{
  \appendix

  \section*{\appendixname}
  \addcontentsline{toc}{section}{\appendixname}

  \titleformat{\section}[block]
  {\normalfont\large\bfseries}
  {\Alph{section}}{1em}{}

  \newtheorem{appendixthm}{Theorem}
  \renewcommand{\theappendixthm}{\arabic{appendixthm}}

  \newtheorem{appendixlem}{Lemma}
  \renewcommand{\theappendixlem}{\arabic{appendixlem}}

  \newtheorem{appendixcor}{Corollary}
  \renewcommand{\theappendixcor}{\arabic{appendixcor}}

  \let\thm\appendixthm
  \let\endthm\endappendixthm
  \let\lem\appendixlem
  \let\endlem\endappendixlem
  \let\cor\appendixcor
  \let\endcor\endappendixcor
}

\theoremstyle{definition}

\newtheorem{con}{Condition}[section]

\newtheorem{lem}{Lemma}[section]
\newtheorem{cor}{Corollary}[section]
\newtheorem{thm}{Theorem}[section]
\newtheorem{remark}{Remark}

\definecolor{orange}{rgb}{1.0, 0.647, 0.0}

\title{Robust $M$-Estimation of Scatter Matrices \\via Precision Structure Shrinkage}
\author[1]{Soma Nikai}
\author[2]{Yuichi Goto}
\author[2]{Koji Tsukuda}
\affil[1]{Joint Graduate School of Mathematics for Innovation, Kyushu University}
\affil[2]{Faculty of Mathematics, Kyushu University}
\date{}

\begin{document} 
\maketitle 

\begin{abstract}
Maronna's and Tyler's $M$-estimators are among the most widely used robust estimators for scatter matrices. 
However, when the dimension of observations is relatively high, their performance can substantially deteriorate in certain situations, particularly in the presence of clustered outliers. 
To address this issue, we propose an estimator that shrinks the estimated precision matrix toward the identity matrix. 
We derive a sufficient condition for its existence, discuss its statistical interpretation, and establish upper and lower bounds for its additive finite sample breakdown point. 
Numerical experiments confirm the robustness of the proposed method. 

\end{abstract}

\section{Introduction}
The estimation of scalar multiples of a covariance matrix and associated quantities, such as inverses, eigenvalues, and eigenvectors, plays an important role in multivariate statistical analysis. 
For example, in principal component analysis, it is important to focus on the eigenvectors and the contribution ratios of a covariance matrix. 
In addition, clustering methods such as spectral clustering algorithms and discriminant methods such as Fisher's linear discriminant analysis also require estimators of scalar multiples of a covariance matrix or its inverse. 
Throughout this paper, we refer to scalar multiples of a covariance matrix as covariance structures that are also called scatter matrices or shape matrices. 
This paper investigates methods for estimating the covariance structure of a population distribution. 

Let us consider a $p$-dimensional population distribution with mean vector $\bm{0}$ and covariance matrix $\Sigma$, from which an i.i.d.\ sample of size $n$ is drawn. 
The Sample Covariance Matrix (SCM for short) is commonly used as an estimator of a covariance matrix. 
However, it may perform poorly when the data contain outliers. 
In fact, the Additive finite sample Breakdown Point (AdBP for short) of the SCM is $1 / (n+1)$, which implies that the inclusion of even one single outlier in the data can severely degrade the performance of the SCM. 
See Subsection~\ref{ssec:AdBP} for the definition of the AdBP. 

To address the presence of outliers, the $M$-estimators of scatter matrices $V$ for elliptical distributions \cite{M76}, in particular Tyler's $M$-Estimator (TME for short) that is distribution-free within the class of elliptical distributions \cite{T87}, have been proposed as classical approaches. 
These methods simultaneously estimate covariance and precision structures through simple algorithms based on the Mahalanobis distance. 
Here, throughout this paper, we refer to the inverse of a covariance matrix as the precision matrix and to scalar multiples of a precision matrix as precision structures. 
The classical $M$-estimators are not only computationally feasible but also enjoy asymptotic properties such as $\sqrt{n}$-consistency and asymptotic normality \cite{M76, T88}. 
In addition, desirable properties regarding asymptotic relative efficiency and influence functions are also often highlighted. 
For these reasons, many extensions and variants of these classical $M$-estimators have been studied \cite{AS07, CWH11, D98, DT05, GLN20, OT14, SW14, TFNO23, TYN23}. 
However, it is known that the asymptotic breakdown point and AdBP of classical $M$-estimators are at most $1/p$ \cite{DT05, M76, T8614}. 

Various robust estimators and algorithms have been proposed in order to tackle low breakdown points of classical $M$-estimators in high dimensional settings. 
These include the Stahel--Donoho (SD- for short) estimator \cite{S81}, the minimum volume ellipsoid estimator and the Minimum Covariance Determinant (MCD- for short) estimator \cite{R85}, the S-estimator \cite{D87}, the MM-estimator \cite{TT00}, as well as methods based on statistical depth functions \cite{CGR18}. 
These methods achieve higher breakdown points but are generally more computationally demanding, especially in high dimensional settings, than the classical $M$-estimators \cite{M76, T87}. 
A comparison between classical $M$-estimators and these methods from different perspectives remains important \cite{T91}, even though algorithmic and implementation level efforts to improve their computational efficiency have been reported \cite{PP01, PP07, Rv99, ZSD23}. 
Notably, all the estimators mentioned above, including the classical $M$-estimators, are linearly equivariant. 
See Remark~\ref{rem:le} for the definition of linear equivariance. 
Advantages of linearly equivariant estimators of $V$ under elliptical distributions lie not only in their explicit equivariance property but also in the applicability of the classical asymptotic theory. 
In the setting where $p$ is fixed and $n$ tends to infinity, the existence of a probability limit implies the consistency of the estimated covariance structure \cite{MMYS17}. 
On the other hand, for linearly equivariant estimators of $V$, it is known that their replacement finite sample breakdown points are at most $(n-p+1)/(2n+p-1)$ when $p<n$ and the data are in general position \cite{D87} and so this upper bound is small when $p$ is close to $n$. 

To date, various $M$-estimation methods for $V$, particularly those derived from the TME, have been proposed \cite{D98, GLN20, SW14}. 
Furthermore, penalized $M$-estimators, in which a penalty term is added to the objective function of the classical $M$-estimation problem, have also been widely studied \cite{OT14, PCQ14, SBP14, TYN23, W12a, ZWG13}. 
Among these approaches, two particularly well-known estimators are the diagonally loaded Normalized Sample Matrix Inversion (LNSMI for short) estimator, also known as the Chen–Wiesel–Hero (CWH for short) estimator, and the Regularized Tyler's $M$-estimator (RegTYL for short) \cite{AS07, CWH11, OT14}. 
The LNSMI and RegTYL are simple estimators that introduce a ridge type shrinkage toward the identity matrix $I_p$ in the covariance structure estimation part of the TME, and they can be interpreted as incorporating the well-known Ledoit--Wolf (LW- for short) estimator \cite{LW04} of a covariance matrix \cite{OT14, W08}. 
In relation to LNSMI and RegTYL, the Spatial Sign Covariance Matrix (SSCM for short) \cite{VKO00} and the Generalized SSCM (GSSCM for short) have recently been revisited by \cite{RR19, TYN23}. 
Although the SSCM sacrifices linear equivariance, the loss of equivariance enables the estimator to exceed the aforementioned upper bound on the breakdown point \cite{D87}. 
When the population mean is known, the asymptotic breakdown point of SSCM reaches one. 
Recent research investigates the relationship between estimators including LNSMI and RegTYL, and GSSCM, and examines sufficient conditions under which the AdBP of LNSMI and RegTYL reaches one \cite{TYN23}. 
These findings suggest that shrinkage toward $I_p$ may lead to estimators with strong robustness, yet its effect on robustness has not been thoroughly investigated. 

Through numerical experiments, we found that the performance of both TME and LNSMI deteriorates in the presence of clustered outliers when $n > p$. 
The reason appears to be a failure of the precision structure. 
Hence, in this paper, we propose an $M$-estimator of $V$ that introduces 
ridge type shrinkage toward $I_p$ in a precision structure estimation step. 
Specifically, we incorporate the Efron--Morris (EM- for short) estimator \cite{EM76}, which is a well-known estimator of the precision matrix along with the Kubokawa--Srivastava (KS- for short) estimator \cite{KS08}, into the $M$-estimation of $V$. 

Clustered outliers have been extensively discussed for many years in the context of breakdown point analysis and numerical studies, and they are regarded as one of the most challenging patterns of outliers for robust estimation of $V$ \cite{A98, MY95, T8614}. 
Although the MCD and SD-estimators possess high theoretical robustness, clustered contamination has been discussed in several applied and computational studies \cite{H01, JP01, PP01, PP07, WR94, RW01}. 
To address these issues, several approaches have been proposed \cite{PP01, PP07}, which include methods that conduct preprocessing prior to computing a robust estimator \cite{JP01}. 
Our proposed estimator, however, can deal with clustered outliers without requiring such preprocessing, provided an appropriately selected shrinkage coefficient. 

We examine the theoretical background of the proposed estimator within the frameworks of penalized $M$-estimation and Bayesian estimation, clarify the conditions for the existence of a solution, and then evaluate the robustness of the proposed estimator in terms of the AdBP. 
This study investigates how robustness is affected when a ridge type shrinkage, which differs from those used in methods such as RegTYL, LNSMI, and the recently proposed Hybrid $M$-Estimator (HME for short) \cite{TYN23}, is incorporated into the $M$-estimation of $V$. 
Focusing on the framework of penalized $M$-estimation, we numerically demonstrate that the proposed method performs comparably to existing methods such as SCM, TME, LNSMI, and SSCM in situations where they operate effectively, while it outperforms them in the presence of clustered outliers. 
In addition, this paper provides a structured synthesis of the relationships among various classes of estimators within the frameworks of penalized $M$-estimation. 

This paper is organized as follows. 
In Section~\ref{preliminaries}, we precisely formulate the problem and describe the existing estimators. 
In Section~\ref{proposal}, we introduce the proposed estimator, explain its theoretical background, and discuss the conditions for the existence of estimators characterized by estimating equations. 
In Section~\ref{robustness}, we evaluate the robustness of the proposed estimator from the perspective of the AdBP. 
In Section~\ref{numerical}, we examine the performance of the proposed estimator under various conditions through numerical experiments. 
The proofs of mathematical statements are given in Section~\ref{sec:proof}. 
Finally, in Section~\ref{summary}, we summarize the contents of this paper. 
Details and additional simulation results are provided in the Appendix. 

\textit{Notations:}
Let $\mathbb{R}_{\geq 0} \coloneqq \{x \in \mathbb{R} : x \geq 0 \}$ and $ \mathbb{R}_{> 0} \coloneqq \{x \in \mathbb{R} : x > 0 \}$. 
For $I \in \mathbb{N}$, let $[I] \coloneqq \{1, \ldots, I\}$. 
Let $\lceil \cdot \rceil$ denote the ceiling function. 
Let $(\cdot)^{\top}$ denote the transpose of a matrix or vector. 
For a nonsingular matrix $A$, let $A^{-1}$ denote its inverse, and let $(A^{-1})^\top$ be abbreviated as $A^{-\top}$. 
For a matrix $A$, let $A^\dagger$ denote the Moore--Penrose pseudoinverse of $A$. 
For a square matrix $A$, let $\mathrm{Tr}[A]$ and $\det(A)$ denote the trace and determinant of $A$, respectively. 
For a matrix $A$, let $\|A\|_\mathrm{F}^{2} \coloneqq \mathrm{Tr}[AA^{\top}]$ denote the squared Frobenius norm of $A$. 
For a symmetric matrix $A$, let $\lambda_{j}(A)$ denote the $j$-th eigenvalue of $A$. 
For symmetric matrices $A$ and $B$, let $A \preceq B$ indicate that $B - A$ is a positive semidefinite symmetric matrix, which is known as the Löwner partial order. 
For $\bm{x} \in \mathbb{R}^p$, let $\|\bm{x}\|$ denote the Euclidean norm of $\bm{x}$. 
Let $\mathbb{S}^{p-1} = \{ \bm{x} \in \mathbb{R}^p : \| \bm{x} \| = 1 \}$ denote the unit sphere. 
Let $\Delta^{p-1} = \{ \bm{\eta} \in \mathbb{R}^p : 0 \leq \eta_j \, (j \in [p]), \, \sum_{j=1}^p \eta_j = 1 \}$ denote the $(p-1)$-dimensional unit simplex. 
Let $\mathcal{N}_p(\bm{0}, \Sigma)$ denote the $p$-dimensional normal distribution with mean vector $\bm{0}$ and covariance matrix $\Sigma$. 
Consider a distribution with mean vector $\bm{0}$ and covariance matrix $\Sigma \in \mathcal{S}_{++}^{p}$. 
For $\bm{x} \in \mathbb{R}^p$, $\bm{x}^\top \Sigma^{-1} \bm{x}$ is called the squared Mahalanobis distance of $\bm{x}$, and $\sqrt{\bm{x}^{\top} \Sigma^{-1} \bm{x}}$ is called the Mahalanobis distance of $\bm{x}$ with respect to $\Sigma$. 
Following the usual convention, for an estimator $\hat{\Sigma}$ of $\Sigma$, $\bm{x}^\top \hat{\Sigma}^{-1} \bm{x}$ is also called the squared Mahalanobis distance of $\bm{x}$ with respect to $\hat{\Sigma}$. 

\section{Preliminaries} \label{preliminaries} 
\subsection{Setup and notation for the probability space} 
Hereafter, we fix a single measurable space $(\Xi, \mathcal{F})$. 
Consider a probability space $(\Xi, \mathcal{F}, \mathcal{P})$ equipped with a probability measure $\mathcal{P}$ and a general measurable space $(S, \mathcal{A})$, and let $\bm{R}: (\Xi, \mathcal{F}) \to (S, \mathcal{A})$ be a random vector or random variable. 
Let us denote by $P_{\bm{R}}(\cdot)$ the measure on $(S, \mathcal{A})$ induced by $\bm{R}$, that is, $P_{\bm{R}}(\cdot) \coloneqq \mathcal{P}(\bm{R}^{-1}(\cdot))$. 
Moreover, for $\chi = \mathbb{R}^p$, we consider the observation space $(\chi, \mathcal{B}(\chi))$. 
For a random vector $\bm{X}: (\Xi, \mathcal{F}) \to (\chi, \mathcal{B}(\chi))$, we write $P_{\bm{X}}(\cdot)$ simply as $P(\cdot)$ when no confusion arises from the context, and define probabilities by $P(\text{condition}) \coloneqq P_{\bm{X}}({\bm{X} \in \chi : \text{condition}})$. 

Let us denote by $\hat{P}_n(\cdot)$ the empirical measure on the $(\chi, \mathcal{B}(\chi))$ based on an i.i.d.\ sample $\bm{X}_1, \dots, \bm{X}_n$ of size $n$. 
Let us denote by $\mathrm{E}_{P_{\bm{X}}}[\cdot]$ the expectation operator with respect to the probability measure $P_{\bm{X}}(\cdot)$, and, when no confusion arises from the context, simply by $\mathrm{E}[\cdot]$. 
Let $\Theta \coloneqq \mathcal{S}_{++}^p$ denote the parameter space of interest, and, unless stated otherwise. 
We consider the Euclidean topology on both the observation space and the parameter space. 

\subsection{Elliptical distributions}
When a continuous random vector $\bm{X} \in \mathbb{R}^{p}$ has a Probability Density Function (PDF for short) of the form 
\begin{equation*}
    f(\bm{x}) = C_{p,g} \det(V^{-1}) \, g(\bm{x}^{\top}V^{-1}\bm{x}) \quad (\bm{x} \in \mathbb{R}^p), 
\end{equation*}
we say that $\bm{X}$ follows a (centered) elliptical distribution, and denote it by $\bm{X} \sim \textnormal{ES}_p(\bm{0}, V, g(\cdot))$. 
The function, $g: \mathbb{R}_{\geq 0} \to \mathbb{R}_{> 0}$ is a nonincreasing function, referred to as the density generator. 
Moreover, $C_{p,g} > 0$ is a normalizing constant depending on $p$ and $g(\cdot)$, and $V \in \mathcal{S}_{++}^p$ is referred to as the scatter matrix or shape matrix. 
An elliptical distribution is characterized by the property that the contours of its PDF are ellipsoidal. 
For example, the multivariate normal, the multivariate $t$, and the multivariate Cauchy distributions are included in the class of elliptical distributions. 
For an elliptical distribution with a finite covariance matrix $\Sigma \in \mathcal{S}_{++}^p$, there exists a constant $c \in \mathbb{R}_{>0}$ such that $\Sigma = c V$. 
Therefore, the scatter matrix $V$ is also referred to as the covariance structure. 

\subsection{Formulation of the problem}
Consider the statistical model $(\Xi, \mathcal{F}, \{\mathcal{P}_\theta\}_{\theta \in \Theta} )$, where $\mathcal{P}_\theta$ denotes the probability measure corresponding to the parameter $\theta \in \Theta$. 

We consider the case where the mean is known and an i.i.d.\ sample $\bm{X}_1, \ldots, \bm{X}_n$ of size $n$ is drawn from the population distribution $\textnormal{ES}_p(\bm{0}, V, g(\cdot))$. 
We assume that $n > p$, and treat the density generator $g(\cdot)$ as an unknown but fixed function. 
Let $\bm{x}_1, \ldots, \bm{x}_n$ denote realizations of $\bm{X}_1, \ldots, \bm{X}_n$. 
We then consider the problem of estimating the scatter matrix $V \in \Theta$ based on the observed data ${\bm{x}_1, \ldots, \bm{x}_n}$. 

\subsection[Penalized M-estimation]{Penalized $M$-estimation} \label{penalizedM}
The MLE of the scatter matrix $V$ is obtained by minimizing 
\begin{equation}\label{eq:lf}
    \mathscr{L}(V) \coloneqq \frac{1}{n} \sum_{i=1}^{n} \rho^{*}\left(\bm{x}_i^{\top} V^{-1} \bm{x}_i\right) - \log \, (\det (V))^{-1} \quad (V \in \mathcal{S}_{++}^p)
\end{equation}
which is the negative log-likelihood divided by $n$. 
The function $\rho^{*}(s) \coloneqq - \log g(s)$ is defined for $s \geq 0$. 

Now, in \eqref{eq:lf}, let us replace $\rho^{*}(\cdot)$ with a \textit{general} function $\rho: \mathbb{R}_{\geq 0} \to \mathbb{R}$, which is not necessarily related to the density generator $g(\cdot)$. 
We then consider
\begin{equation}\label{eq:mf}
\mathscr{M}(V) \coloneqq \frac{1}{n}\sum_{i=1}^{n} \rho\left(\bm{x}_i^{\top}V^{-1}\bm{x}_i\right)-\log \, (\det (V))^{-1} \quad (V \in \mathcal{S}_{++}^p)
\end{equation}
and seek $V \in \mathcal{S}_{++}^p$ that minimizes the objective function~\eqref{eq:mf}. 
In this case, the minimizer $V$ of \eqref{eq:mf} is an $M$-estimator \cite{M76, OT14, T87}. 
For example, TME corresponds to the choice $\rho(s) = p \log s$ ($s > 0$) in \eqref{eq:mf} \cite{T87}. 
In this paper, no assumptions are imposed on the function $\rho(\cdot)$ itself. 
Instead, assumptions are placed on the function corresponding to its first derivative, as stated in Condition~\ref{conm}. 

The optimization of the objective function obtained by adding a penalty term
\begin{equation}\label{eq:pmf}
\mathscr{M}_{\alpha}(V) \coloneqq \frac{1}{n}\sum_{i=1}^n \rho\left(\bm{x}_i^{\top}V^{-1}\bm{x}_i\right)-\log \, (\det (V))^{-1} + \mathcal{P}(V ; \, \alpha) \quad (V \in \mathcal{S}_{++}^p) 
\end{equation}
has been studied \cite{DT16, FHT08, OT14, TYN23, W08, W12a}. 
This framework is referred to as penalized $M$-estimation, and the parameter $\alpha$ is called the penalty or shrinkage coefficient. 
For example, LNSMI and RegTYL correspond to the estimation obtained from \eqref{eq:pmf} with $\rho(s) = p \, \log \, s$ ($s > 0$) and $\mathcal{P}(V ; \, \alpha) = \alpha \mathrm{Tr}[V^{-1}]$ \cite{OT14}. 

\subsection{Characterization of the estimator} 
We assume that the loss function $\rho(\cdot)$ is differentiable, and define $w(\cdot) \coloneqq \rho'(\cdot)$. 
In this case, a stationary point of \eqref{eq:mf} is a solution to the estimating equation 
\begin{equation}\label{eseq:m}
  V = \frac{1}{n} \sum_{i=1}^n w(\bm{x}_i^\top V^{-1} \bm{x}_i) \bm{x}_i \bm{x}_i^\top \textnormal{.}
\end{equation}
In \eqref{eq:pmf}, when $\mathcal{P}(V ; \, \alpha) = \alpha \mathrm{Tr}[V^{-1}]$, a stationary point is a solution to the estimating equation
\begin{equation}\label{eseq:penam}
  V = \frac{1}{n} \sum_{i=1}^n w(\bm{x}_i^\top V^{-1} \bm{x}_i) \bm{x}_i \bm{x}_i^\top + \alpha I_p \textnormal{.}
\end{equation}
A solution satisfying the estimating equation~(\ref{eseq:m}) is linearly equivariant (see Remark \ref{rem:le}), whereas a solution satisfying (\ref{eseq:penam}) is not, and this difference affects the robustness of the estimator \cite{D87}. 
In \eqref{eq:pmf}, when the Kullback--Leibler type penalty 
\begin{align*}
    \mathcal{P}(V ; \alpha) &= \textnormal{KL}(\mathcal{N}_p(\bm{0}, V) \Vert \mathcal{N}_p(\bm{0}, I_p)) \\
    &= \alpha \{ \mathrm{Tr}[V^{-1}] + \log \det (V) \}
\end{align*}
with $\alpha \in (0,\infty)$ is considered, the estimating equation is obtained as 
\begin{equation}\label{eseq:penam2}
  V = \frac{1}{1 + \alpha} \cdot \frac{1}{n} \sum_{i=1}^n \frac{\bm{x}_i \bm{x}_i^\top}{\bm{x}_i^\top V^{-1} \bm{x}_i} + \frac{\alpha}{1 + \alpha} I_p \textnormal{,}
\end{equation}
where $\mathrm{KL}(\mu \Vert \nu)$ is the Kullback--Leibler divergence of a distribution $\mu$ from a distribution $\nu$ \cite{vH14}. 

The $M$-estimator of $V$ is often defined through estimating equations such as (\ref{eseq:m}), (\ref{eseq:penam}), and (\ref{eseq:penam2}),
and computed using methods based on fixed-point algorithms for the corresponding estimating equations. 
Estimators characterized as solutions to the estimating equations~(\ref{eseq:m}), (\ref{eseq:penam}), and (\ref{eseq:penam2}) do not necessarily possess scale equivariance. 
See Remark~\ref{rem:le} for the definition of scale equivariance. 

Various assumptions on the loss function $\rho(\cdot)$ and the weight function $w(\cdot)$ have been considered \cite{KT91, M76, OT14, ZWG13}. 
Following \cite{M76}, we impose the conditions (A)--(E) below on the weight function $w: \mathbb{R}_{\geq 0} \to \mathbb{R}_{\geq 0}$, and also impose the condition (F) below. 

\begin{con}\label{conm}
    \textbf{}

    \begin{enumerate}[label= (\Alph*)]
        \item The weight function $w: \mathbb{R}_{\geq 0} \to \mathbb{R}_{\geq 0}$ is continuous and nonincreasing. 
        \item For $\psi(s) \coloneqq s w(s)$ ($s \geq 0$), it holds that $\kappa \coloneqq \sup_{s \ge 0} \{\psi(s)\} < \infty$. 
        \item The function $\psi: \mathbb{R}_{\geq 0} \to \mathbb{R}_{\geq 0}$ is nondecreasing. 
        \item The function $\psi: \mathbb{R}_{\geq 0} \to \mathbb{R}_{\geq 0}$ is strictly increasing in $\{s \in \mathbb{R}_{\geq 0} : \psi(s) < \kappa \}$. 
        \item It holds that $\kappa > p$. 
        \item For any $(p-1)$-dimensional linear subspace $H \subset \mathbb{R}^p$, it holds that $P(H) < 1 - p / \kappa$. 
    \end{enumerate}
\end{con}

\begin{remark}
Condition~\ref{conm}~(A)--(D) is satisfied for
  Huber-type weight functions, 
  \begin{equation*}
    w(s) =
    \left\{
      \begin{array}{ll}
      1 & \quad (0 \leq s \leq c), \\
      c/s & \quad (s > c), 
  \end{array}
  \right.
  \end{equation*}
  for $c > 0$, the $t$-distribution type weight functions $w(s) = (p + \nu) / (s + \nu)$ ($s \geq 0$) for $\nu \in \mathbb{N}$, and the Tyler type weight function $w(s) = p / s$ ($s \geq 0$). 
\end{remark}

\subsection{Additive finite sample breakdown point} \label{ssec:AdBP} 
The breakdown point is one of the measures of an estimator's robustness against contamination by outliers in the data \cite{HR09, M76, MMYS17, T8614, TYN23}. 
Several definitions of the breakdown point exist, which can be broadly classified into those in the asymptotic setting and those in the finite sample setting. 
In this paper, we focus on the breakdown point in the finite sample setting. 

We consider a situation in which $n$ data points $\mathrm{X} = \{\boldsymbol{x}_1, \ldots, \boldsymbol{x}_n \}$ are observed, together with arbitrary $m$ data points $\mathrm{Y} = \{ \boldsymbol{y}_1, \ldots, \boldsymbol{y}_m \}$. 
Then, the contaminated dataset is defined as $\mathrm{Z} \coloneqq \mathrm{X} \cup \mathrm{Y}$. 
Typically, $\mathrm{Y}$ corresponds to the outlier observations. 
Moreover, let an estimator of $\theta$ be denoted by $\hat{\theta}(\cdot)$, and the estimate of $\theta$ computed from a dataset $\mathrm{W}$ be denoted by $\hat{\theta}(\mathrm{W})$. 
Furthermore, let $B: \Theta^2 \to \mathbb{R}$ denote a function measuring the discrepancy between $\hat{\theta}(\mathrm{X})$ and $\hat{\theta}(\mathrm{Z})$. 
A notion of the function $B(\cdot, \cdot)$, called the discrepancy function, will be explained later. 

For the outlier ratio
\[
    \varepsilon_m \coloneqq \frac{m}{n+m}, 
\] 
a fixed estimator $\hat{\theta}(\cdot)$, and a fixed dataset $\mathrm{X}$, we define 
\begin{equation*}
  b \left(\varepsilon_m ; \, \hat{\theta}(\cdot), \mathrm{X} \right) \coloneqq
    \left\{
      \begin{array}{ll}
      \displaystyle \sup_\mathrm{Y} B \left(\hat{\theta} (\mathrm{X}), \hat{\theta}(\mathrm{Z}) \right) & \quad (\mathrm{Z} \in \mathcal{E}_m(\mathrm{X}) ), \\
      +\infty & \quad (\mathrm{Z} \notin \mathcal{E}_m(\mathrm{X}) ),
\end{array}
\right.
\end{equation*}
where $\mathcal{E}_m(\mathrm{X}) = \{ \mathrm{Z} = \mathrm{X} \cup \mathrm{Y} : \, \hat{\theta}(\mathrm{Z}) \textnormal{ exists in } \Theta \, (\mathrm{X}: \text{fixed})\}$. 

Let an estimator of $V$ be denoted by $\hat{V}(\cdot)$, and the estimate of $V$ computed from a dataset $\mathrm{W}$ be denoted by $\hat{V}(\mathrm{W})$. 
We say that the estimator $\hat{V}(\cdot)$ breaks down with respect to the dataset $\mathrm{X}$ if $b \left(\varepsilon_m ; \, \hat{V}(\cdot), \mathrm{X} \right) = + \infty$. 
The additive finite sample breakdown point (AdBP) of the estimator $\hat{V}(\cdot)$ is defined by 
\begin{equation*}
  \varepsilon^* \left( \hat{V}(\cdot), \mathrm{X} \right) \coloneqq \min \left\{ \varepsilon_m ; \, b \left(\varepsilon_m ; \, \hat{V}(\cdot), \mathrm{X} \right) = \infty \right\}. 
\end{equation*}

We consider a discrepancy function $B(V_1, V_2)$ that satisfies
\[
    B \left( V_1, V_2 \right) \to \infty \quad \text{as } \lambda_p(V_2) \to 0 \ \text{or}\ \lambda_1(V_2) \to \infty.
\]
Examples of such discrepancy functions include $\mathrm{Tr}[V_1 V_2^{-1} + V_1^{-1} V_2]$ and $\| \log (V_1^{-1/2} V_2 V_1^{-1/2}) \|_{\mathrm{F}}$ \cite{T8614, TYN23}. 
Here, for $X \in \mathcal{S}_{++}^p$, we define
\[
    \log X \coloneqq Q \, \mathrm{diag} \left( \log (\lambda_1(X)), \dots, \log (\lambda_p(X)) \right) \, Q^\top, 
\]
where $Q \, \mathrm{diag} \left( \lambda_1(X), \dots, \lambda_p(X) \right) \, Q^\top$ is the eigenvalue decomposition of $X$ with an orthogonal matrix $Q$. 

\begin{remark} \label{rem:le}
  In this remark, let $\mathrm{X}$ be a set of arbitrary $p$-dimensional vectors $\bm{x}_1,\ldots,\bm{x}_n$. 
  Let us denote by $A\mathrm{X}$ the dataset obtained by transforming each data point of $\mathrm{X}$ with a $p \times p$ nonsingular matrix $A$, i.e., $\bm{x}_i \mapsto A \bm{x}_i$ for $i \in [n]$. 
  An estimator $\hat{V}(\cdot)$ is said to be linearly equivariant if, for any $p \times p$ nonsingular matrix $A$, the relation $\hat{V}(A\mathrm{X}) = A \hat{V}(\mathrm{X}) A^\top$ holds, whenever $\hat{V}(\mathrm{X})$ and $\hat{V}(A\mathrm{X})$ exist in $\mathcal{S}_{++}^p$. 
  Moreover, an estimator $\hat{V}(\cdot)$ is said to be scale equivariant if, for any nonzero scalar $a$, the relation $\hat{V}(a \mathrm{X}) = a^2 \hat{V}(\mathrm{X})$ holds. 
\end{remark} 

\section{Proposed method} \label{proposal}
\subsection{Limitations of existing methods}
We examine the performance of TME and LNSMI through numerical experiments under two types of outlier configurations. 
We observe that when outliers form clusters, the performance of TME and LNSMI deteriorates significantly for a sample size of $N \coloneqq n + m = 100$ and an outlier ratio of $3\%$, when the data dimension exceeds $p \approx 35$. 
See Section~\ref{numerical} for the detailed settings and Figure~\ref{fig:r1}-(b, d, f) for the results. 
Moreover, this phenomenon is found to be closely related to the breakdown point of TME. 
Based on the results of the numerical experiments, we conjecture that this phenomenon arises from the inaccurate estimation of the precision structure; see Appendix~A. 
To address this issue, we propose an algorithm that applies a more direct treatment of the precision structure compared with LNSMI. 

\subsection{Numerical algorithm of the proposed method}
Given an initial value $V_0 \in \mathcal{S}_{++}^p$ and a penalty parameter $\alpha \in [0,1]$, we propose the following iterative algorithm: 
\begin{align*}
  \Omega_{k+1} &\leftarrow (1 - \alpha) V_k^{-1} + \alpha I_p, \\
  \tilde{V}_{k+1} &\leftarrow \frac{1}{n}\sum_{i=1}^n w(\bm{x}_i^\top \Omega_{k+1} \bm{x}_i) \bm{x}_i \bm{x}_i^\top, \\
  V_{k+1} &\leftarrow \frac{p}{\mathrm{Tr}[\tilde{V}_{k+1}]} \tilde{V}_{k+1}. 
\end{align*}
The quantities $V_k$ and $\Omega_k$ denote the estimated covariance structure and precision structure, respectively, computed at the $k$-th iteration. 
The specific procedure for determining the value of the penalty parameter is described in Subsection~\ref{coef}. 

The proposed algorithm iterates the following three steps until a prescribed stopping criterion is satisfied. 
First, we estimate $\Omega$ by applying a shrinkage step of the same form as that used in the EM-type estimator, given $V_k$. 
Next, we update $V_{k+1}$ in an $M$-estimation manner, using the resulting $\Omega_{k+1}$. 
Then, we perform trace normalization so that $\mathrm{Tr}[V_{k+1}] = p$ holds. 
For the determination of the shrinkage coefficient and the stopping criterion, see Remark~\ref{rem:tc} and Subsection~\ref{coef}. 

While LNSMI and RegTYL shrink the estimated covariance structure toward $I_p$ at each iteration, the proposed method shrinks the estimated precision structure toward $I_p$ at each iteration. 
In addition, the proposed method differs from HME in the way it shrinks the precision structure toward $I_p$ and in whether trace normalization is applied. 

\subsection[Interpretation of the proposed method as a penalized M-estimator]{Interpretation of the proposed method as a penalized $M$-estimator}
We investigate the relationship between the proposed method and the framework of penalized $M$-estimation. 
In this subsection, we focus on the case where $\mathcal{P}(V ; \, \alpha) = - \log \left|\det \left(I_{p} - \alpha V \right)\right|$ in \eqref{eq:pmf}, and consider the framework in which we minimize
\begin{equation}\label{eq:pmf20}
    \mathscr{M}_{\alpha}^{\textnormal{(P)}}(\Omega) \coloneqq \frac{1}{n}\sum_{i=1}^{n} \rho(\bm{x}_{i}^{\top} \Omega \bm{x}_{i}) - \log \, \det (\Omega) - \log \left|\det \left(I_{p} - \alpha \Omega^{-1}\right)\right| \quad (\Omega \in \mathcal{S}_{++}^p) 
\end{equation}
as a function of $\Omega = V^{-1}$. 
The stationary point of the objective function \eqref{eq:pmf20} satisfies 
\begin{equation}\label{eseqpro0}
\Omega = \left( \frac{1}{n}\sum_{i=1}^{n} w(\bm{x}_i^\top \Omega \bm{x}_i) \bm{x}_i \bm{x}_i^\top \right)^{-1} + \alpha I_{p}
\end{equation}
when $\det (I_p - \alpha \Omega^{-1}) \neq 0$. 
Equation~(\ref{eseqpro0}) represents the relation formally satisfied by the convergence point of the proposed algorithm for an appropriate choice of $\alpha$ and $w(\cdot)$. 
Furthermore, from the derivation of \eqref{eseqpro0}, it follows that $I_p - \alpha \Omega^{-1} \in \mathcal{S}_{++}^p$. 

Moreover, in \eqref{eq:pmf20}, if we set $w(s) = 1 / \beta$ $(s \geq 0)$ for $\beta > 0$, the stationary point is given by
\begin{equation}\label{eq:EM}
\hat{\Omega} = \beta S^{-1} + \alpha I_p \text{,}
\end{equation}
where $S$ denotes the SCM. 
The class of estimators in the form of \eqref{eq:EM} corresponds to the well-known class of EM estimators for precision matrix estimation \cite{EM76}. 

Under the constraint $\mathrm{Tr}[V] = c$ for a constant $c$, the term $- \log \left|\det \left(I_{p} - \alpha V \right)\right|$ is minimized when $V = (c / p) I_p$. 
This result is consistent with the shrinkage direction employed in the proposed algorithm. 

In \eqref{eq:pmf}, when the R\'enyi type penalty 
\begin{align*}
    \mathcal{P}(V ; \alpha) &= \frac{2(\alpha -1)}{\alpha} \textnormal{R}_{\alpha}(\mathcal{N}_p(\bm{0}, V) \Vert \mathcal{N}_p(\bm{0}, I_p)) \\
    &= -\frac{1}{\alpha} \log \left( \det(V)^{\alpha - 1} \det\left( \alpha I_p + (1-\alpha) V \right) \right), 
\end{align*}
with $\alpha \in (0,\infty)$ is considered, the estimating equation is obtained as 
\begin{equation*}
  \Omega = \frac{1}{\alpha} \left\{ \frac{1}{n} \sum_{i=1}^n w(\bm{x}_i^\top \Omega \bm{x}_i) \bm{x}_i \bm{x}_i^\top \right\}^{-1} + \frac{\alpha - 1}{\alpha} I_p 
\end{equation*}
where $\textnormal{R}_{\alpha}(\mu \Vert \nu)$ is the R\'enyi divergence of a distribution $\mu$ from a distribution $\nu$ \cite{vH14}. 

\subsection{Characterization of the proposed method by its estimating equations}\label{eseqpro}
By substituting $V = (\Omega - \alpha I_p)^{-1}$ into the estimating equation~(\ref{eseqpro0}), we obtain the estimating equation 
\begin{equation} \label{eseq:pro1}
  V = \frac{1}{n} \sum_{i=1}^n w\left( \bm{x}_i^\top \left\{V^{-1} + \alpha I_p \right\} \bm{x}_i \right) \bm{x}_i \bm{x}_i^\top \text{.}
\end{equation} 
Alternatively, substituting $V^{-1} = (\Omega - \alpha I_p) / (1 - \alpha)$ into the estimating equation~(\ref{eseqpro0}) leads to the more direct form
\begin{equation} \label{eseq:pro2}
  V = \frac{1}{n} \sum_{i=1}^n w\left( \bm{x}_i^\top \left\{(1-\alpha) V^{-1} + \alpha I_p \right\} \bm{x}_i \right) \bm{x}_i \bm{x}_i^\top \text{.}
\end{equation}
The proposed method can be characterized by the estimating equation~(\ref{eseqpro0}) in terms of $\Omega$ or by (\ref{eseq:pro1}) and (\ref{eseq:pro2}) in terms of $V$.  

When $\alpha \searrow 0$, the estimating equation~(\ref{eseq:pro2}) coincides with~(\ref{eseq:m}), namely, the ordinary $M$-estimator. 
On the other hand, when $\alpha = 1$, the resulting estimator corresponds to the class of GSSCMs \cite{RR19}. 
In particular, when $\alpha = 1$ and $w(s) = 1 / s$ ($s > 0$), the SSCM is obtained as a special case. 

The estimators characterized by the estimating equations~(\ref{eseq:pro1}) and~(\ref{eseq:pro2}) are not linearly equivariant, but they are orthogonally equivariant. 
That is, for any $p \times p$ orthogonal matrix $Q$, whenever $\hat{V}(\mathrm{X})$ and $\hat{V}(Q\mathrm{X})$ exist in $\mathcal{S}_{++}^p$, the relation $\hat{V}(Q\mathrm{X}) = Q \hat{V}(\mathrm{X}) Q^\top$ holds. 

\subsection{Interpretation of the proposed method as a Bayesian estimator}%\label{Bayes}
We discuss the relationship between the proposed method and the maximum a posteriori estimation (MAP for short). 
Let $\bm{D} = (D_1, D_2, \ldots, D_p)^{\top}$ be a random vector taking values in $\Delta^{p-1}$, and suppose that $\bm{D} \sim \textnormal{Dir}_{p}(\gamma_1)$.
The notation $\textnormal{Dir}_{p}(\gamma_1)$ denotes the $p$-dimensional symmetric Dirichlet distribution with parameter $\gamma_1 > 0$. 
For $\gamma_2 > 0$, the random vector $\bm{\Lambda} = (\Lambda_1, \Lambda_2, \ldots, \Lambda_p)^{\top} \in \mathbb{R}^p$ obtained through the linear transformation
\begin{equation*}
    \Lambda_j = -\frac{1}{\gamma_2} D_j + \frac{1}{\gamma_2} \qquad (j \in [p])
\end{equation*}
has the PDF
\begin{equation*}
  f_{\bm{\Lambda}}(\bm{\lambda}; \gamma_1, \gamma_2) = \frac{\Gamma\left(p \gamma_1 \right)}{\left(\Gamma(\gamma_1)\right)^p}  \cdot \gamma_2^p \cdot \prod_{k=1}^p \left| 1 - \gamma_2 \lambda_k \right|^{\gamma_1 - 1} \quad \left( \bm{\lambda} = (\lambda_1, \ldots, \lambda_p)^{\top} \in \Delta_2 \right) \text{,}
\end{equation*}
where the support is $\Delta_2 \coloneqq \{\bm{\eta} \in \mathbb{R}^p : 0 \leq \eta_k \leq 1 / \gamma_2 \, (k \in [p]), \sum_k \eta_k = (p-1) / \gamma_2\}$. 

Let $P_{\bm{\Lambda}}^{(\gamma_1, \gamma_2)}$ denote the distribution induced by the probability density function $f_{\bm{\Lambda}}(\cdot ; \, \gamma_1, \gamma_2)$. 
We consider the hierarchical model given by
\begin{align*}
    &\bm{X} \, | \, V, \gamma_1, \gamma_2 \sim \textnormal{ES}_p(\bm{0}, V; g) \text{,}\\
    &\left(\lambda_1(V), \lambda_2(V), \ldots, \lambda_p(V)\right)^{\top} \, | \, \gamma_1, \gamma_2 \sim P_{\bm{\Lambda}}^{(\gamma_1, \gamma_2)} \text{.}
\end{align*}
That is, we place the prior $P_{\bm{\Lambda}}^{(\gamma_1, \gamma_2)}$ only on the eigenvalues of the scatter matrix $V$. 
Under this model, the negative of the log posterior density is given by
\begin{equation*}
    \sum_{i=1}^{n} \rho^{*} \left(\bm{x}_{i}^{\top} V^{-1} \bm{x}_{i}\right) - n \log \, (\det (V))^{-1} - \log \, \prod_{k=1}^p \left|1 - \gamma_2 \lambda_k(V) \right|^{\gamma_1 - 1} + \textnormal{(constant)} \text{.}
\end{equation*}
Focusing on the case $\gamma_1 = n+1$ and, as in Subsection~\ref{penalizedM}, allowing $\rho(\cdot)$ to be a \textit{general} function not necessarily related to the density generator $g(\cdot)$, the problem of maximizing the log posterior reduces to the minimization of the objective function~\eqref{eq:pmf20}. 

The condition $\gamma_2 > 0$ in the prior distribution $P_{\bm{\Lambda}}^{(n + 1, \gamma_2)}$ ensures that the largest eigenvalue of the scatter matrix $V$ remains finite.

% \begin{remark}
%   The case $\gamma_1 \neq n+1$ is left for future work. 
% \end{remark}

\subsection[Existence in the positive definite cone]{Existence of the proposed estimator in $\mathcal{S}_{++}^p$}
In this subsection, we discuss the existence of $V \in \mathcal{S}_{++}^p$ that satisfies the following equations
\begin{equation}
  V = \mathrm{E}_{P_{\bm{X}}} \left[ w\left( \bm{X}^\top \left\{V^{-1} + \alpha I_p \right\} \bm{X} \right) \bm{X} \bm{X}^\top \right] \label{eseq:prex1} 
\end{equation} 
and
\begin{equation}
  V = \mathrm{E}_{P_{\bm{X}}} \left[ w\left( \bm{X}^\top \left\{ ( 1- \alpha) V^{-1} + \alpha I_p \right\} \bm{X} \right) \bm{X} \bm{X}^\top \right] \label{eseq:prex2}
\end{equation}
that are expressed in a more general form than the estimating equations~(\ref{eseq:pro1}) and~(\ref{eseq:pro2}). 

First, our next theorem establishes the existence of a solution $V \in \mathcal{S}_{++}^p$ to \eqref{eseq:prex1}. 

\begin{thm}\label{thm:prexcon1}
Let Condition~\ref{conm}-(A)--(F) hold.
Then there exists a solution to \eqref{eseq:prex1} in $\mathcal{S}_{++}^p$. 
\end{thm}

\begin{remark}
  The proof of Theorem~\ref{thm:prexcon1} shows that the sequence of covariance structures obtained from the iteration without the trace normalization step is monotonically nonincreasing with respect to the Löwner partial order. 
  This explains why the proposed algorithm incorporates a trace normalization step. 
  Without trace normalization, the scale of the estimated covariance structure (for example, its trace) may approach zero, which could make the shrinkage factor practically meaningless. 
  To avoid this issue, the proposed algorithm conducts trace normalization. 
\end{remark}

By considering the empirical measure in \eqref{eseq:prex1}, we obtain the following. 

\begin{cor}\label{cor:prexcon}
Let Condition~\ref{conm}-(A)--(D) and~(F) hold.
Then there exists a solution to the estimating equation~(\ref{eseq:pro1}) in $\mathcal{S}_{++}^p$, provided $\kappa > n(p-1) / (n-p)$.
\end{cor}

Moreover, as a necessary condition for the existence of a solution to the estimating equation~(\ref{eseq:pro1}) in $\mathcal{S}_{++}^p$, we obtain the following lemma. 

\begin{lem}\label{lem:prexnescon}
Let Condition~\ref{conm}-(B) and~(C) hold. 
If there exists a solution to the estimating equation~(\ref{eseq:pro1}) in $\mathcal{S}_{++}^p$, then 
  \begin{description}
    \item[(1)] $p \leq \kappa$;
    \item[(2)] $\hat{P}_n(\{\bm{0}\}) \leq 1 - p / \kappa$. 
  \end{description}
\end{lem}

Theorem~\ref{thm:prexcon1}, Corollary~\ref{cor:prexcon}, and Lemma~\ref{lem:prexnescon} can be derived in the same manner for \eqref{eseq:prex2} and \eqref{eseq:pro2}. 

\section{Robustness of the proposed estimator} \label{robustness}
In this section, we discuss the robustness of the estimator $\hat{V}(\cdot)$, characterized as a solution to the estimating equation~\eqref{eseq:pro1}, from the perspective of AdBP. 
Recall that $\varepsilon_m = m/(n+m)$ denotes the outlier ratio and that $\mathcal{E}_m(\mathrm{X}) = \{ \mathrm{Z} = \mathrm{X} \cup \mathrm{Y} : \, \hat{V}(\mathrm{Z}) \textnormal{ exists in } \Theta \, (\mathrm{X}: \text{fixed})\}$. 
The following lemma provides sufficient conditions for the existence or non-existence of $\hat{V}(\mathrm{Z})$
under contamination. 

\begin{lem}\label{exbp}
  \textbf{}

  \noindent
  \begin{description}
      \item[(1)] Let Condition~\ref{conm}-(A)--(D) and~(F) hold and suppose that $\kappa > n(p-1) / (n - p)$. 
      If $\varepsilon_m < 1 - p / \kappa - (p-1) / n$, then $\mathrm{Z} \in \mathcal{E}_m(\mathrm{X})$. 
      \item[(2)] Let Condition~\ref{conm}-(B) and~(C) hold and suppose that $\kappa > p$. 
      If $\varepsilon_m > 1 - p / \kappa$, then $\mathrm{Z} \notin \mathcal{E}_m(\mathrm{X})$. 
  \end{description}
\end{lem}

The next lemma provides sufficient conditions on $\varepsilon_m$
under which $\lambda_1(\hat{V}(\mathrm{Z}))$ and $1/\lambda_p(\hat{V}(\mathrm{Z}))$ are bounded above, provided that $\hat{V}(\mathrm{Z})$ exists. 

\begin{lem}\label{bdbp}
  \textbf{}

  \noindent
  \begin{description}
      \item[(1)] Let Condition~\ref{conm}-(A)--(C) hold. If $\varepsilon_m = 1$, then the set $\{ \lambda_1(\hat{V}(\mathrm{Z})) ; \, \mathrm{Z} \in \mathcal{E}_m(\mathrm{X}) \}$ is bounded above. 
      \item[(2)] Let Condition~\ref{conm}-(B),~(C),~and~(E) hold. If $\varepsilon_m < 1 - p / \kappa$, then the set $\{ 1 / \lambda_p(\hat{V}(\mathrm{Z})) ; \, \mathrm{Z} \in \mathcal{E}_m(\mathrm{X}) \}$ is bounded above. 
  \end{description}
\end{lem}

\begin{remark}
  Lemma~\ref{bdbp}-(1) provides a desirable property of the proposed estimator. 
  The proposed estimator does not \textit{breakdown from above}, in the sense that its largest eigenvalue  remains bounded whenever $\hat{V}(\mathrm{Z})$ exists. 
\end{remark}

\begin{remark}
  For non-penalized $M$-estimators, a stringent condition $\varepsilon_m < 1 / \kappa$ is sufficient to obtain the same conclusion as in Lemma~\ref{bdbp}-(1) \cite{T8614}. 
  In addition, for RegTYL, LNSMI, and HME, the condition $\kappa < 1$ is sufficient for the same conclusion \cite{TYN23}.
  However, the condition $\kappa < 1$ is incompatible with the necessary condition $\kappa \geq p$ for the existence of a non-penalized $M$-estimator in $\mathcal{S}_{++}^p$. 
  Hence, RegTYL, LNSMI, and HME involve a trade-off between bias and robustness \cite{TYN23}. 
  In contrast, the proof for the proposed estimator shows that the much milder condition $\kappa > 1$ is sufficient. 
\end{remark}

Lemmas~\ref{exbp} and~\ref{bdbp} yield the following theorem, which evaluates the AdBP $\varepsilon^* ( \hat{V}(\cdot), \mathrm{X} )$ of the estimator $\hat{V}(\cdot)$ defined as the solution to the estimating equation~\eqref{eseq:pro1}. 

\begin{thm}\label{bpbp}
Let Condition~\ref{conm}-(A)--(D) and~(F) hold and suppose that $\kappa > n(p-1) / (n - p)$.
  Then it holds that
  \begin{equation*}
    1 - \frac{p}{\kappa} - \frac{p-1}{n} \le \varepsilon^* \left( \hat{V}(\cdot), \mathrm{X} \right) \le 1 - \frac{p}{\kappa}.
  \end{equation*}
\end{thm}

\begin{remark}
  The upper bound $1 - p / \kappa$ of the AdBP in Theorem~\ref{bpbp} is attained when there exists a point mass at the origin. 
\end{remark}

\section{Numerical experiments} \label{numerical}
In this section, we compare the finite sample performance of the proposed estimator, which uses the Tyler type weighting function $w(s) = p / s \, (s > 0)$, with that of the SCM, TME, LNSMI, SSCM, and $I_p$. 

\subsection{Settings}\label{sec:simusets}
We describe the settings of the numerical experiments in three parts. 

\medskip 

\begin{enumerate}[label=\textbf{$\langle \arabic* \rangle$}, left=0pt]

    \item \textbf{Generation of covariance matrices} \\ 
    Fix the data dimension $p$, specify the trace $\mathrm{Tr}[\Sigma]$ of the true scatter matrix $\Sigma \in \mathcal{S}_{++}^{p}$, the largest contribution ratio $c_1 \coloneqq \lambda_1 (\Sigma) / \mathrm{Tr}[\Sigma]$, and the condition number $c_2 \coloneqq \lambda_1 (\Sigma) / \lambda_{p}(\Sigma)$, e.g., $\mathrm{Tr}[\Sigma] = p$, $c_1 =0.30$, $c_2=50$, and generate $\Sigma = \Sigma(c_1, c_2)$ that satisfies these specifications. 
    Since all methods compared here possess orthogonal equivariance as discussed in Subsection~\ref{eseqpro}, we may choose the eigenvectors of $\Sigma$ as the standard basis. 
    
    \item \textbf{Generation of main body data} \\ 
    Specify all sample sizes $N$ and outlier ratios $\xi$, e.g., $N = 100$, $\xi = 0.03$, and set the number of outliers as $m = \lceil N \xi \rceil$. 
    Generate the main body data $\bm{z}_i$ for $i \in [n]$, where $n \coloneqq N - m$ from the multivariate normal distribution with mean vector $\bm{0}$ and scatter matrix $\Sigma$. 

    \item \textbf{Generation of outliers} \\ 
    Generate the outliers $\bm{z}_{i}$ for $i \in [N] \setminus [n]$ so that their squared Mahalanobis distances equal a constant multiple $k$, e.g., $k = 10$, of the sample mean of the squared Mahalanobis distances of the main body data points
    \[
        r \coloneqq \frac{1}{n} \sum_{i=1}^{n} \bm{z}_i^\top \Sigma^{-1} \bm{z}_i, 
    \]
    namely, the outliers satisfy
    \[
        \bm{z}_i^\top \Sigma^{-1} \bm{z}_i = kr \quad (i \in [N] \setminus [n]) \text{.}
    \]
    From the set of points satisfying this condition as illustrated in Figure~\ref{fig:od}, generate the outliers by the following two methods: 
    \begin{enumerate}[label={}, left=0pt]
        \item 1.\ \textbf{unclustered outliers} 
        
        We obtain $m$ points from random directions, shown in Figure~\ref{fig:ot}-(a);
        \item 2.\ \textbf{clustered outliers} 
        
        We obtain $m$ points from a single random direction, as a \textit{cluster}, shown in Figure~\ref{fig:ot}-(b). 
        Specifically, outliers are randomly generated from $\mathcal{N}_p(\bm{\mu}_\mathrm{o}, 0.01 I_p)$, where $\bm{\mu}_\mathrm{o}$ is given by the method $\langle 3 \rangle$. 
    \end{enumerate}
\end{enumerate}

\begin{figure}[t]
    \centering
    \includegraphics[width=0.5\columnwidth]{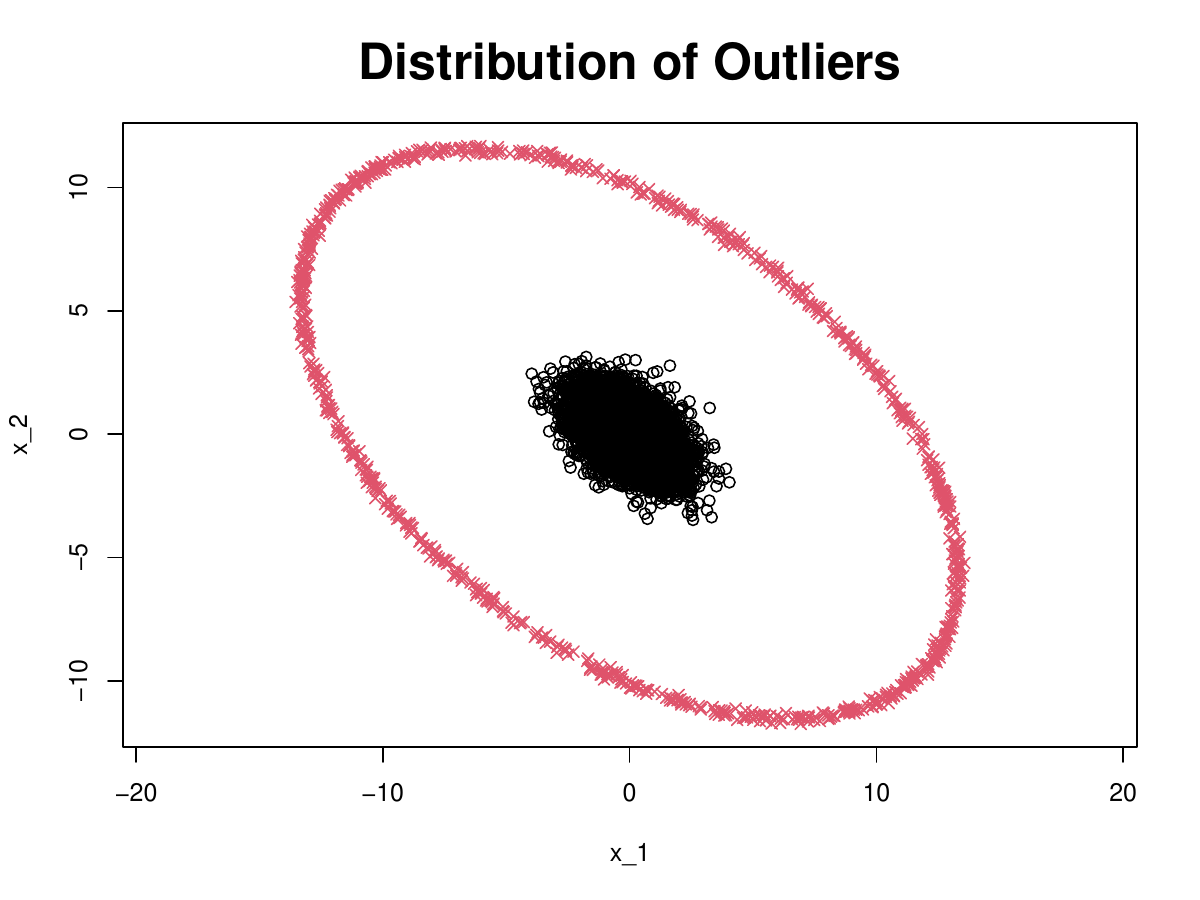}
    \caption{An illustration of the distribution of outliers for the case $p=2$, with $c_1 = 0.6$, $c_2 = 1000$, $N = 10000$, $\xi = 0.1$, and $k = 10$. 
    Data points belonging to the main body are plotted as black circles,
    and outliers are represented by red crosses. }
    \label{fig:od}
\end{figure}

\begin{figure}[t]
    \centering
    \includegraphics[width=\columnwidth]{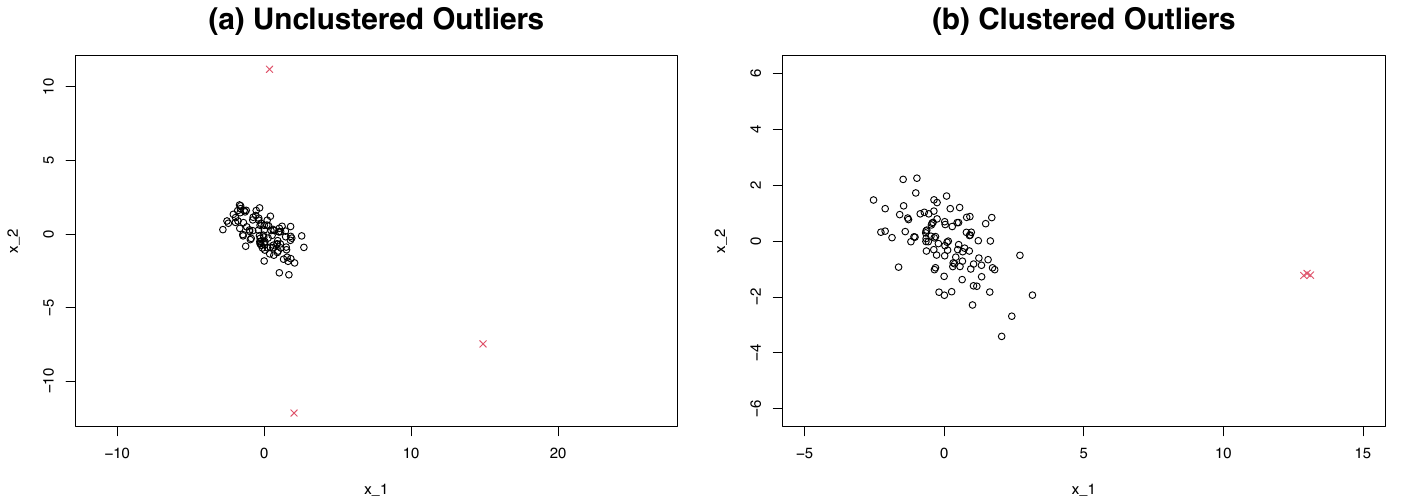}
    \caption{An illustration of the outlier generation methods for the case $p = 2$, with $c_1 = 0.6$, $c_2 = 1000$, $N = 100$, $\xi = 0.03$, and $k = 10$. 
    The left panel illustrates unclustered outliers, while the right panel illustrates clustered outliers.
    Data points belonging to the main body are plotted as black circles, and outliers are represented by red crosses. }
    \label{fig:ot}
\end{figure}

The performance of SCM, TME, LNSMI, SSCM, $I_p$, and the proposed Tyler type estimator is evaluated under each condition by varying 
\begin{itemize}
    \item[(i)]
     data generation parameters $(p, \mathrm{Tr}[\Sigma], c_1, c_2, N, \xi, k)$,
    \item[(ii)] 
      the generating distribution of the main body data,
    \item[(iii)]
     the outlier configurations (unclustered outliers or clustered outliers).
\end{itemize}
The performance of each estimator is assessed using the RMSE computed over 1000 trials for each setting.
In addition, the shrinkage coefficient of LNSMI is determined based on Theorem~2 in \cite{CWH11}. 

First, Subsection~\ref{sec:simr1} examines how the estimation error changes with varying dimensions of the random vector.
These results demonstrate that TME and LNSMI fail to perform properly under certain conditions. 
Next, Subsection~\ref{simuresult2} examines how the estimation error changes when the shrinkage coefficient $\alpha$ is varied. 
This analysis reveals the behavior of LNSMI and the proposed Tyler type estimator, regardless of whether the shrinkage coefficient is successfully computed. 
In this analysis, we examine the effect of the shrinkage coefficient on the estimator. 

\begin{remark}
  Once the largest contribution ratio $c_1$, the condition number $c_2$, and the trace $\mathrm{Tr}[\Sigma]$ are specified, the largest and smallest eigenvalues $\lambda_1 (\Sigma)$ and $\lambda_{p}(\Sigma)$ are uniquely determined. 
  Details on the generation of $\lambda_2(\Sigma), \ldots, \lambda_{p-1}(\Sigma)$ can be found in Appendix~C. 
\end{remark}

\begin{remark} \label{rem:tc}
  The initial values for TME, LNSMI, and the proposed method are set to the trace normalized SCM. 
  For TME and LNSMI, the stopping criterion is defined as 
  \begin{equation*}
      \left\| V_{k+1} - V_k \right\|_\mathrm{F} < \epsilon,
  \end{equation*}
  with $\epsilon = 10^{-6}$. 
  For the proposed method, among several stopping criteria examined, the following stopping criterion yielded the best performance and is therefore adopted: 
  \begin{equation*} 
    \min \left\{ \left\| \Omega_{k+1} - \Omega_k \right\|_\mathrm{F}, \left\| V_{k+1} - V_k \right\|_\mathrm{F} \right\} < \epsilon \text{.}
  \end{equation*}
\end{remark}

\subsection{Penalty coefficient computation for the proposed method}\label{coef}
Following \cite{CWEH10, CWH11, LW04, OT14}, we determine the shrinkage coefficient $\hat{\alpha}$ by considering the optimization problem 
\begin{equation*}
    \underset{\alpha \in \mathbb{R}}{\operatorname{arg\,min}} \ \mathrm{E} \left[\left\| (1 - \alpha) C_\textnormal{P}^{-1} + \alpha I_p - \Omega \right\|_\mathrm{F}^{2} \right] \text{, where} \; C_\textnormal{P} \coloneqq \frac{p}{N}\sum_{i=1}^N \frac{\bm{x}_{i} \bm{x}_{i}^{\top}}{\bm{x}_{i}^{\top} \Omega \bm{x}_{i}} \text{.}
\end{equation*}
The solution to this optimization problem is given by
\begin{equation}\label{rho2o}
  \alpha_\textnormal{o} \coloneqq \frac{\mathrm{Tr}[\Omega] - \mathrm{Tr}\left[\mathrm{E}\left[C_\textnormal{P}^{-1}\right]\right] - \mathrm{Tr}\left[\mathrm{E}\left[C_\textnormal{P}^{-1}\right]\Omega\right] + \mathrm{E}\left[\left\| C_\textnormal{P}^{-1} \right\|_\mathrm{F}^{2}\right]}{p-2\mathrm{Tr}\left[\mathrm{E}\left[C_\textnormal{P}^{-1}\right]\right] + \mathrm{E}\left[\left\|C_\textnormal{P}^{-1}\right\|_\mathrm{F}^{2}\right]} \text{.}
\end{equation}
Since $\Omega$ is the true precision structure, $\alpha_\textnormal{o}$ is an oracle quantity. 
Accordingly, by substituting into (\ref{rho2o}) an appropriate estimator $M_2$ for $\Omega$, that is, by defining
\begin{equation*}
    \hat{C}_\textnormal{P} \coloneqq \frac{p}{N}\sum_{i=1}^N \frac{\bm{x}_{i} \bm{x}_{i}^{\top}}{\bm{x}_{i}^{\top} M_2 \bm{x}_{i}}
\end{equation*}
we wish to propose to use
\begin{equation*}
    \frac{\mathrm{Tr}[M_2] - \mathrm{Tr}\left[\mathrm{E}\left[\hat{C}_\textnormal{P}^{-1}\right]\right] - \mathrm{Tr}\left[\mathrm{E}\left[\hat{C}_\textnormal{P}^{-1}\right] M_2 \right] + \mathrm{E}\left[\left\|\hat{C}_\textnormal{P}^{-1}\right\|_\mathrm{F}^{2}\right]}{p - 2\mathrm{Tr}\left[\mathrm{E}\left[\hat{C}_\textnormal{P}^{-1}\right]\right] + \mathrm{E}\left[\left\|\hat{C}_\textnormal{P}^{-1}\right\|_\mathrm{F}^{2}\right]} 
\end{equation*}
as $\hat{\alpha}$. 
However, because it is often difficult to express $\mathrm{E}[\hat{C}_\textnormal{P}^{-1}]$ and $\mathrm{E}[\|\hat{C}_\textnormal{P}^{-1} \|_\mathrm{F}^{2}]$ in closed form, we therefore replace these expectations by their numerical approximations obtained via the bootstrap, and use the resulting plug-in value as $\hat{\alpha}$. 
Specifically, by substituting the bootstrap approximations $m_1$ and $m_2$ for $\mathrm{E}[\hat{C}_\textnormal{P}^{-1}]$ and $\mathrm{E}[ \| \hat{C}_\textnormal{P}^{-1} \|_\mathrm{F}^{2}]$, respectively, we propose the shrinkage coefficient
\begin{equation*}
    \hat{\alpha} \coloneqq \frac{\mathrm{Tr}[M_2] - \mathrm{Tr}\left[ m_1 \right] - \mathrm{Tr}\left[ m_1 M_2 \right] + m_2}{p - 2\mathrm{Tr}\left[ m_1 \right] + m_2} \text{.}
\end{equation*}

\begin{remark}
When approximating $\mathrm{E}[\hat{C}_\textnormal{P}^{-1}]$ and $\mathrm{E}[ \|\hat{C}_\textnormal{P}^{-1} \|_\mathrm{F}^{2}]$ via the bootstrap method, it is necessary to resample as many distinct data points as possible, up to a maximum of $N$, in order to avoid rank deficiency in $\hat{C}_\textnormal{P}$. 
Specifically, when the sample size is $N = 100$, we recommend performing approximately 370 resampling iterations. 
The justification for this choice is provided in Appendix~B. 
\end{remark}

\subsection{Performance across different dimensions}\label{sec:simr1}
\begin{figure}[t]
    \centering
    \includegraphics[width=0.85\columnwidth]{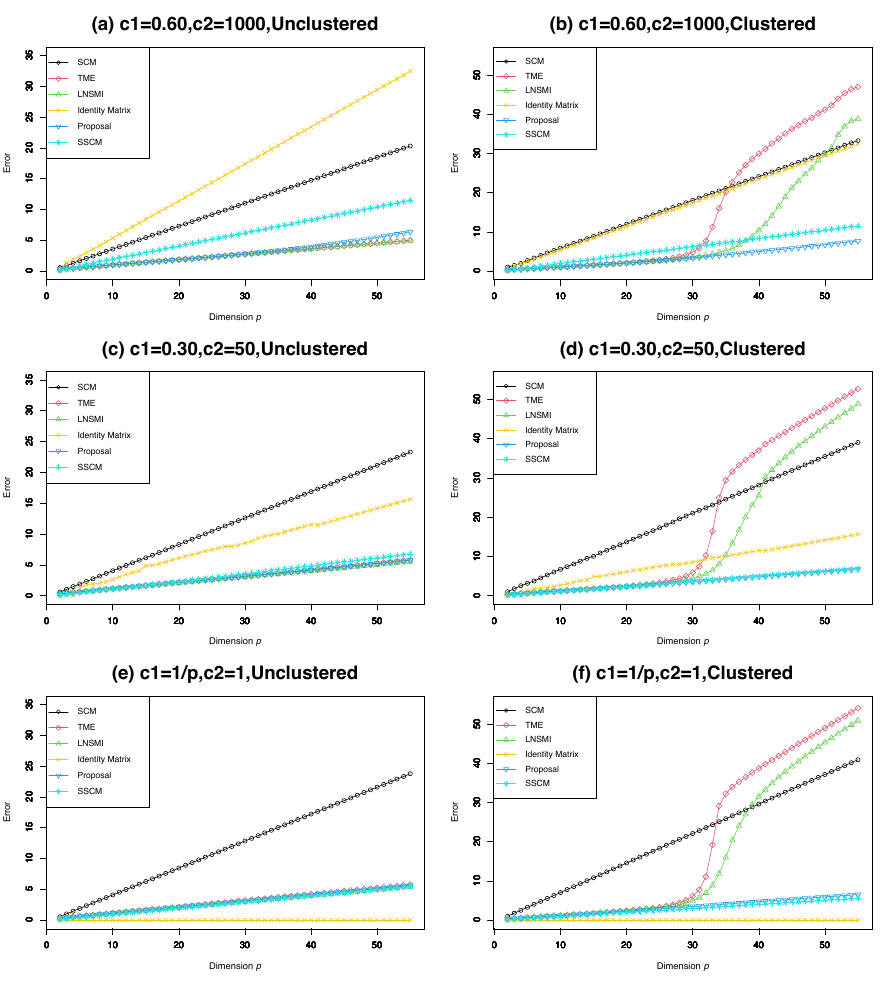}
    \caption{
    Plots of the RMSE values of each estimator for different data dimensions $p$, with lines connecting the plotted points. 
    The horizontal axis represents the dimension $p$ $( = 2, \ldots, 55 )$, and the vertical axis represents the RMSE. 
    The other settings are $N = 100$, $\xi = 0.03$, and $k = 10$. 
    Figure~(a, c, e) show the results for unclustered outliers, while Figure~(b, d, f) correspond to clustered outliers. 
    In Figure~(a, c, e), the curves representing TME, LNSMI, SSCM, and the proposed Tyler type estimator nearly overlap. 
    Furthermore, in Figure~(d, f), the curves representing SSCM and the proposed Tyler type estimator nearly overlap. 
    }\label{fig:r1}
\end{figure}

Figure~\ref{fig:r1} illustrates how the RMSE of each estimator changes with the dimension $p$ of the data points. 
In this subsection, we set $p \in [55] \setminus \{1\}$, $\mathrm{Tr}[\Sigma] = p$, $N = 100$, $\xi = 0.03$, and $k = 10$. 
Numerical values of the estimation errors are provided in Appendix~D.
Additional experimental results under different parameter settings are partially presented in Appendices~E and F. 

\noindent
\textbf{1. unclustered outliers.}

Figure~\ref{fig:r1}-(a, c, e) shows the results for the case of unclustered outliers. 
Specifically, Figure~\ref{fig:r1}-(a), (c), and (e) correspond to
$(c_1, c_2) = (0.6, 1000)$, $(0.3, 50)$, and $(1/p, 1)$, respectively,
where the last case corresponds to $\Sigma = I_p$. 

The curves representing the performance of TME, LNSMI, SSCM, and the proposed Tyler type estimator nearly overlap with one another. 
However, in Figure~\ref{fig:r1}-(a), where the true covariance structure deviates from $I_p$, SSCM exhibits relatively poor estimation accuracy. 
On the other hand, the behavior of all estimators remains highly stable, and in particular, TME, LNSMI, SSCM, and the proposed Tyler type estimator show better estimation accuracy than SCM. 

\noindent
\textbf{2. clustered outliers.}

Figure~\ref{fig:r1}-(b, d, f) shows the results for the case of clustered outliers. 
Specifically, Figure~\ref{fig:r1}-(b), (d), and (f) correspond to
$(c_1, c_2) = (0.6, 1000)$, $(0.3, 50)$, and $(1/p, 1)$, respectively. 

We observe that when the outlier ratio is $\xi = 0.03$, the performance of TME and LNSMI deteriorates markedly once the dimension exceeds $p = 35$, becoming even worse than that of SCM. 
In contrast, SSCM and the proposed Tyler type estimator successfully suppress the severe degradation seen in TME and LNSMI. 
These phenomena are consistently observed under various settings with different values of $\mathrm{Tr}[\Sigma]$, $N$, $\xi$, and $k$. 
It is known that the asymptotic breakdown point of TME lies between $1/(p+1)$ and $1/p$, and the observed behavior of TME conforms to this property \cite{DT05}. 

\subsection{Sensitivity to shrinkage coefficients}\label{simuresult2}
\begin{figure}[t]
    \centering
    \includegraphics[width=0.85\columnwidth]{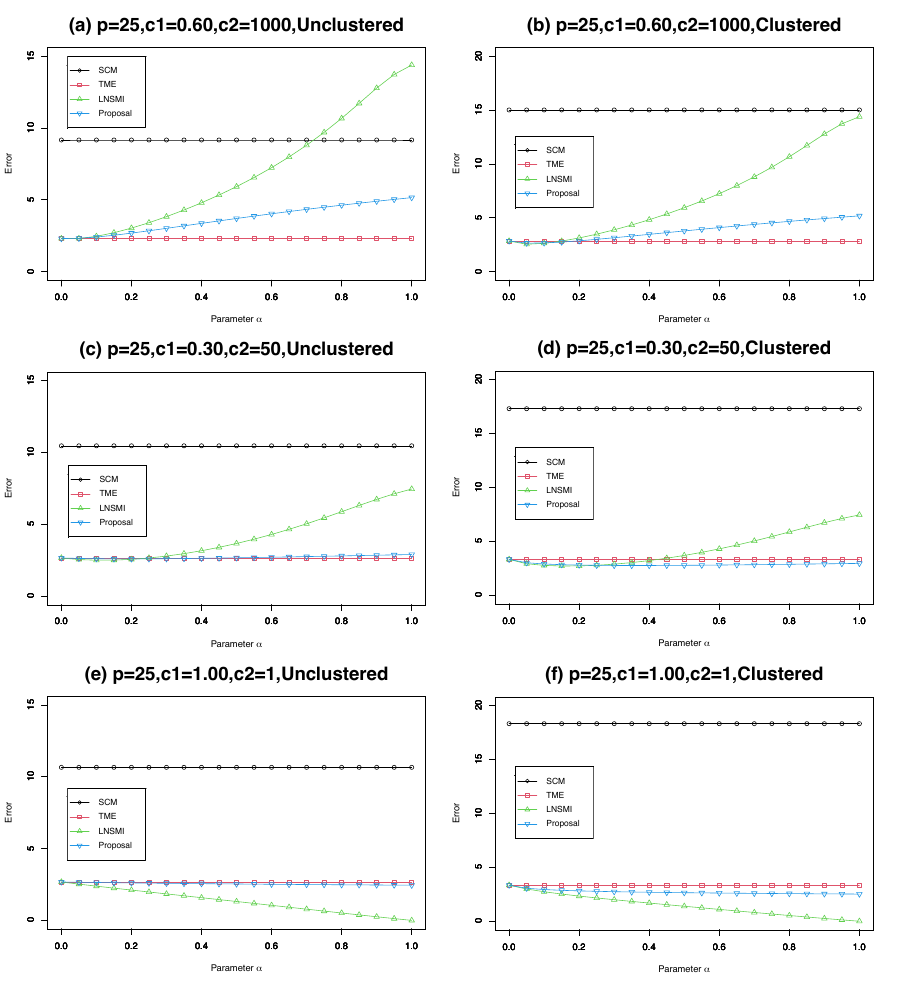}
    \caption{
    Plots of the RMSE values of each estimator for different values of the shrinkage coefficient $\alpha$ in the case $p = 25$, with lines connecting the plotted points. 
    The horizontal axis represents the shrinkage coefficient $\alpha$, and the vertical axis represents the RMSE. 
    The other settings are $N = 100$, $\xi = 0.03$, and $k = 10$. 
    Figure~(a, c, e) show the results for unclustered outliers, while Figure~(b, d, f) correspond to clustered outliers. 
    }
    \label{fig:r21}
\end{figure}

\begin{figure}[t]
    \centering
    \includegraphics[width=0.85\columnwidth]{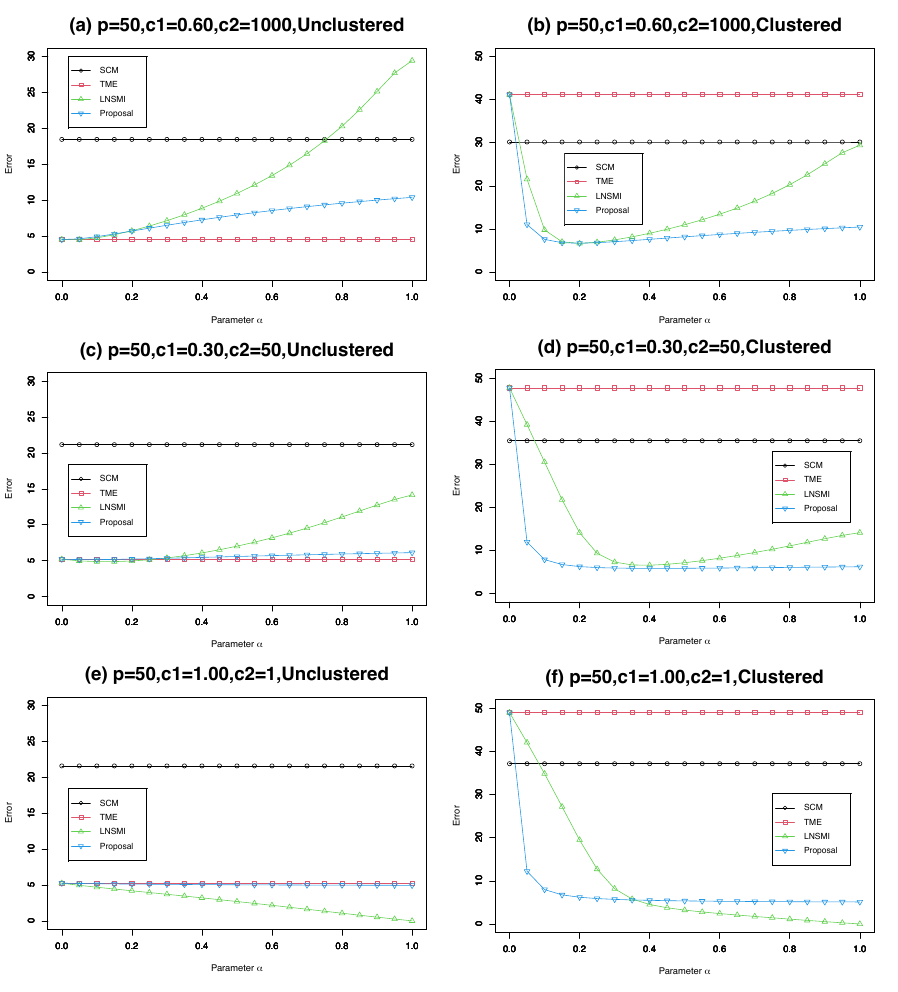}
    \caption{
    Plots of the RMSE values of each estimator for different values of the shrinkage coefficient $\alpha$ in the case $p = 50$, with lines connecting the plotted points. 
    The horizontal axis represents the shrinkage coefficient $\alpha$, and the vertical axis represents the RMSE. 
    The other settings are $N = 100$, $\xi = 0.03$, and $k = 10$. 
    Figure~(a, c, e) show the results for unclustered outliers, while Figure~(b, d, f) correspond to clustered outliers. 
    }
    \label{fig:r22}
\end{figure}

To examine the effect of the shrinkage coefficient, we set $\alpha \in [0,1]$ to various values for LNSMI and the proposed Tyler type estimator and compare the behavior of each estimator. 
We set $p \in \{25, 50\}$, $\mathrm{Tr}[\Sigma] = p$, $N = 100$, $\xi = 0.03$, and $k = 10$. 
Figure~\ref{fig:r21} presents the results for $p = 25$, and Figure~\ref{fig:r22} shows those for $p = 50$. 
In Figures~\ref{fig:r21} and \ref{fig:r22}, panels (a, c, e) correspond to the case of unclustered outliers, while panels (b, d, f) correspond to the case of clustered outliers. 
Moreover, in Figures~\ref{fig:r21} and \ref{fig:r22}, panels (a, b), (c, d), and (e, f) correspond to $(c_1, c_2) = (0.6, 1000)$, $(0.3, 50)$, and $(1/p, 1)$, respectively. 

Since SCM, TME, and SSCM do not involve shrinkage, their corresponding results remain constant. 
For LNSMI, when $\alpha = 1$, the output becomes $I_p$.
Therefore, in Figure~\ref{fig:r21}-(e, f) and Figure~\ref{fig:r22}-(e, f), where the true covariance structure is $I_p$, the RMSE value of LNSMI becomes $0$ at $\alpha = 1$. 
In contrast, for the proposed Tyler type estimator, the result coincides with SSCM when $\alpha = 1$. 
The estimation errors of TME and those of LNSMI and the proposed Tyler type estimator show very similar behavior when $\alpha \approx 0$.
Indeed, when $\alpha \approx 0$, both LNSMI and the proposed Tyler type estimator operate in a manner similar to TME.
In particular, LNSMI and the proposed Tyler type estimator theoretically coincide when $\alpha = 0$. 

Below, we provide an overview of the behavior of the estimators as the shrinkage coefficient $\alpha$ is varied.
Note that the meaning of nonzero values of $\alpha$ differs between LNSMI and the proposed method, so the results of the two methods cannot be directly compared when $\alpha \neq 0$. 

First, comparing the points at which each estimator attains its smallest estimation error over the range $\alpha \in [0,1]$, we see that the proposed Tyler type estimator achieves accuracy comparable to the other estimators in many settings. 
On the other hand, in the presence of clustered outliers, LNSMI can also attain estimation accuracy similar to that of the proposed Tyler type estimator if an appropriate value of $\alpha$ is chosen. 
However, except for the cases in Figure~\ref{fig:r21}-(e, f) and Figure~\ref{fig:r22}-(e, f), where the true covariance structure is $I_p$, the curve representing the proposed Tyler type estimator generally lies below that of LNSMI. 
This observation suggests that the proposed Tyler type estimator is robust with respect to the choice of $\alpha$. 

\section{Proofs}\label{sec:proof}
\subsection{Proof of Theorem~\ref{thm:prexcon1}}
The main idea of the proof of Theorem~\ref{thm:prexcon1} is essentially the same as the proof by Lemma 1 and Theorem 1 of Maronna~\cite{M76}. 
We first state the following lemma which will be used to prove Theorem~\ref{thm:prexcon1}.
\begin{lem}\label{lem4}
    Let Condition~\ref{conm}-(B) and (C) hold. 
    There exists $r_0 > 0$ such that
    \begin{equation*}
        \mathrm{E} \left[w \left( \left(\frac{1}{r} + \alpha \right) \| \bm{X} \|^2 \right) \bm{X} \bm{X}^\top \right] \preceq r I_p, 
    \end{equation*}
    for any $r \geq r_0$. 
\end{lem}

\begin{proof}[Proof of Lemma~\ref{lem4}]
    \textnormal{}
It suffices to show that
    \begin{equation*}
        \lim_{r \to \infty} \sup_{\| \bm{z} \| = 1} \bm{z}^\top \mathrm{E} \left[w \left( \left(\frac{1}{r} + \alpha \right) \| \bm{X} \|^2 \right) \bm{X} \bm{X}^\top \right] \bm{z} / r = 0.
    \end{equation*}
It holds that
    \begin{align*}
        & \sup_{\| \bm{z} \| = 1} \bm{z}^\top \mathrm{E} \left[w \left( \left(\frac{1}{r} + \alpha \right) \| \bm{X} \|^2 \right) \bm{X} \bm{X}^\top \right] \bm{z} / r \\
        &= \sup_{\| \bm{z} \| = 1} \bm{z}^\top \mathrm{E} \left[\psi \left( \left(\frac{1}{r} + \alpha \right) \| \bm{X} \|^2 \right) \frac{\bm{X} \bm{X}^\top}{(1 / r + \alpha) \| \bm{X} \|^2} \right] \bm{z} / r \\
        &= \sup_{\| \bm{z} \| = 1} \bm{z}^\top \mathrm{E} \left[\psi \left( \left(\frac{1}{r} + \alpha \right) \| \bm{X} \|^2 \right) \frac{\bm{X} \bm{X}^\top}{(1 + \alpha r) \| \bm{X} \|^2} \right] \bm{z} \\
        &\leq \frac{\kappa}{1 + \alpha r}\sup_{\| \bm{z} \| = 1} \bm{z}^\top \mathrm{E} \left[\frac{\bm{X} \bm{X}^\top}{\| \bm{X} \|^2} \right] \bm{z} \\
        &\to 0
    \end{align*}
as $r \to \infty$.
This completes the proof.
\end{proof}

We proceed to show Theorem \ref{thm:prexcon1}. 

\begin{proof}[Proof of Theorem~\ref{thm:prexcon1}]
Let
\begin{equation*}
    F_{\alpha} (V) \coloneqq \mathrm{E} \left[ w\left( \bm{X}^\top \{V^{-1} + \alpha I_p \} \bm{X} \right) \bm{X} \bm{X}^\top \right].
\end{equation*}
We begin by showing that if $V$ is positive definite, then $F_{\alpha} (V)$ is also positive definite. 
To obtain a contradiction, suppose that the assertion, which is $F_{\alpha} (V)$ is positive definite, is false. 
Then there exists $\bm{z} \neq \bm{0}$ such that $\mathrm{E} \left[ w\left( \bm{X}^\top \{V^{-1} + \alpha I_p \} \bm{X} \right) (\bm{z}^\top \bm{X})^2 \right] = 0$, and so 
\begin{equation*}
    P \left( w\left( \bm{X}^\top \{V^{-1} + \alpha I_p \} \bm{X} \right) (\bm{z}^\top \bm{X})^2 = 0 \right) = 1
\end{equation*}
which is equivalent to
\begin{equation*}
    P \left( w\left( \bm{X}^\top \{V^{-1} + \alpha I_p \} \bm{X} \right) = 0 \right) + P \left( \bm{z}^\top \bm{X} = 0 \right) - P(J) = 1,
\end{equation*}
where $J = \left\{\bm{X} \in \mathbb{R}^p ; \, w\!\left( \bm{X}^\top \{V^{-1} + \alpha I_p \} \bm{X} \right) = 0 \,\right\} \cap \left\{\bm{X} \in \mathbb{R}^p ; \, \bm{z}^\top \bm{X} = 0 \,\right\}$.
We see from Condition~\ref{conm}-(C) and (D) that $P(w\left( \bm{X}^\top \{V^{-1} + \alpha I_p \} \bm{X} \right) = 0) = 0$, hence that
\begin{equation*}
        P \left( \bm{z}^\top \bm{X} = 0 \right) - P(J) = 1,
\end{equation*}
and finally that $P(\bm{z}^\top \bm{X} = 0) = 1$. 
This contradicts Condition~\ref{conm}-(F). 

By Lemma~\ref{lem4}, there exists $r_0 > 0$ such that $F_{\alpha}(V_0) \preceq r_0 I_p$.
Let $V_0 \coloneqq r_0 I_p$. 
We define the sequence $\{V_k\}_{k \in \mathbb{N}}$ recursively by $V_{k+1}=F_{\alpha}(V_{k})$. 
This sequence is well defined, and $V_{k+1} \preceq V_{k}$ holds for all $k \in \mathbb{N}$ because $w(\cdot)$ is non-increasing. 
Hence, $\tilde{V} \coloneqq \lim_{k \to \infty} V_{k}$ exists. 

What is left is to show that $\tilde{V}$ is positive definite.
Suppose the assertion is false.
That is, there exists $\bm{z} \neq \bm{0}$ such that $\bm{z}^\top \tilde{V} \bm{z} = 0$.
From Condition \ref{conm}-(B) and (C), for any $b \in (0,\kappa)$, there exists $d_b$ such that $\psi(d) \geq \kappa - b$ for $d \geq d_b$. 
Let $C_{b,k} \coloneqq \{ \bm{x} ; \, \bm{x}^\top (V_k^{-1} + \alpha I_p) \bm{x} < d_b \}$ and $C_b \coloneqq \cap_k C_{b,k}$. 
Then $C_{b,k+1} \subseteq C_{b,k}$ for all $k \in \mathbb{N}$.
Consider a decomposition $V_k = S_k^\top S_k$.
For any $k \in \mathbb{N}$ and $\bm{x} \in C_{b,k}$, it holds that
\begin{align*}
    (\bm{z}^\top \bm{x})^2 &= \{ \left( S_k\bm{z} \right)^\top \left( S_k^{-\top} \bm{x} \right) \}^2 \\
    &\leq (\bm{z}^\top V_k \bm{z})(\bm{x}^\top V_k^{-1} \bm{x}).
\end{align*}
Since $\bm{z}^\top V_k \bm{z} \to 0$ as $k \to \infty$ and $\bm{x}^\top V_k^{-1} \bm{x} < d_b$, any $\bm{x}$ satisfying $\bm{x}^\top V_k^{-1} \bm{x} < d_b$ belongs to $H\coloneqq \{ \bm{x} ; \, \bm{z}^\top \bm{x} = 0 \}$.
Therefore, $C_b \subset H$. 
Consequently, for any $k \in \mathbb{N} \cup \{0\}$ and any $b \in (0,\kappa)$, it holds that
\begin{equation*}
    V_k \succeq F_{\alpha}(V_k) \succeq \mathrm{E} \left[ I(\bm{X} \in C_{b,k}^c) ( \kappa - b) \frac{\bm{X} \bm{X}^\top}{\bm{X}^\top (V_k^{-1} + \alpha I_p) \bm{X}} \right]. 
\end{equation*}
Premultiplying $S_k^{-\top}$, postmultiplying $S_k^{-1}$ and taking the trace, we obtain
\begin{equation*}
    p \geq \mathrm{E} \left[ I(\bm{X} \in C_{b,k}^c) ( \kappa - b) \frac{\bm{X}^\top V_k^{-1} \bm{X}}{\bm{X}^\top (V_k^{-1} + \alpha I_p) \bm{X}} \right] 
\end{equation*}
for any $k \in \mathbb{N} \cup \{0\}$ and any $b \in (0,\kappa)$. 
Fatou's lemma implies that
\begin{align}
    & \liminf_{k \to \infty} \mathrm{E} \left[ I(\bm{X} \in C_{b,k}^c) ( \kappa - b) \frac{\bm{X}^\top V_k^{-1} \bm{X}}{\bm{X}^\top (V_k^{-1} + \alpha I_p) \bm{X}} \right] \notag\\
    &\geq (\kappa - b) \mathrm{E} \left[ \liminf_{k \to \infty} I(\bm{X} \in C_{b,k}^c) \frac{\bm{X}^\top V_k^{-1} \bm{X}}{\bm{X}^\top (V_k^{-1} + \alpha I_p) \bm{X}} \right] \notag \\
    &\geq (\kappa - b) \mathrm{E} \left[ \liminf_{k \to \infty} I(\bm{X} \in H^c) \frac{\bm{X}^\top V_k^{-1} \bm{X}}{\bm{X}^\top (V_k^{-1} + \alpha I_p) \bm{X}} \right]. \label{tochu}
\end{align}
It holds that
\begin{equation*}
    \lim_{k \to \infty}\frac{\bm{X}^\top V_k^{-1} \bm{X}}{\bm{X}^\top (V_k^{-1} + \alpha I_p) \bm{X}} = 
    \begin{cases}
      1 & \text{($\bm{X} \in  \textnormal{Im} ( \tilde{V} )^c$)}, \\
      \frac{1}{1 + \alpha \bm{X}^\top \bm{X} / \bm{X}^\top \tilde{V}^\dagger \bm{X}} & \text{($\bm{X} \in \textnormal{Im} ( \tilde{V} ) \setminus \{\bm{0}\}$)} \\
      0 & \text{($\bm{X} = \bm{0}$)}, 
    \end{cases},
\end{equation*}
because 
\begin{align*}
    \lim_{k \to \infty}\frac{\bm{X}^\top V_k^{-1} \bm{X}}{\bm{X}^\top (V_k^{-1} + \alpha I_p) \bm{X}} &= \lim_{k \to \infty} \frac{1}{1 + \alpha \bm{X}^\top \bm{X} / \bm{X}^\top V_k^{-1} \bm{X}} \\
    &= \lim_{k \to \infty} \frac{1}{1 + \alpha \bm{X}^\top \bm{X} / \left\{ \sum_{j=1}^p (1 / \lambda_j(V_k)) \cdot \left( \bm{X}^\top \bm{\gamma}_j(V_k) \right)^2 \right\}}
\end{align*}
for $\bm{X} \neq \bm{0}$.
Since $\bm{z}^\top \tilde{V} \bm{z} = 0$ yields $\bm{z} \in \textnormal{Ker} (\tilde{V})$ and $\textnormal{Im}(\tilde{V}) = \textnormal{Ker}(\tilde{V})^\perp$ holds, we have $\bm{x}^\top \bm{z} = 0$ for any $\bm{x} \in \textnormal{Im}(\tilde{V})$.
From this, it follows that $\textnormal{Im}(\tilde{V}) \subset H$, which is equivalent to $H^c \subset \textnormal{Im}(\tilde{V})^c$.
From what has already been proved, 
\begin{align*}
    \eqref{tochu} &= (\kappa - b) \mathrm{E} \left[ I(\bm{X} \in H^c) \right] \\
    &= (\kappa - b) (1 - P (H)),
\end{align*}
that is,
\begin{equation*}
    p \geq \kappa (1 - P (H)),
\end{equation*}
which contradicts Condition~\ref{conm}-(F). 

This completes the proof. 
\end{proof}

\begin{remark}
The proof for \eqref{eseq:prex2} is similar. 
\end{remark}

\subsection{Proof of Corollary~\ref{cor:prexcon}}
\begin{proof}[Proof of Corollary~\ref{cor:prexcon}]
Since the population distribution is continuous, any $(p-1)$ data points almost surely lie in a $(p-1)$-dimensional linear subspace. 
Hence, $\hat{P}_n(H) \leq (p-1) / n$ almost surely.
From Condition~\ref{conm}-(F), it follows that $(p-1) / n < 1 - p / \kappa$, and hence $np / (n-p+1) < \kappa$. 
\end{proof}

\subsection{Proof of Lemma~\ref{lem:prexnescon}}
The main strategy for proving Lemma~\ref{lem:prexnescon} is essentially the same as that adopted in Tyler~\cite{T88}. 

\begin{proof}[Proof of Lemma~\ref{lem:prexnescon}]
  \textbf{}

  \begin{description}
    \item[(1)] 
    Premultiplying $V^{-1}$ to the estimating equation \eqref{eseq:pro1} and taking the trace yield
      \begin{equation*}
        p = \frac{1}{n} \sum_{i=1}^n w\left( \bm{x}_i^\top \left\{V^{-1} + \alpha I_p \right\} \bm{x}_i \right) \bm{x}_i^\top V^{-1} \bm{x}_i. 
      \end{equation*}
    This gives
      \begin{align}
        p &= \frac{1}{n} \sum_{i=1}^n \psi\left( \bm{x}_i^\top \left\{V^{-1} + \alpha I_p \right\} \bm{x}_i \right) \frac{\bm{x}_i^\top V^{-1} \bm{x}_i}{\bm{x}_i^\top \left\{V^{-1} + \alpha I_p \right\} \bm{x}_i} \label{nesto0}\\
        &\leq \frac{1}{n} \sum_{i=1}^n \psi\left( \bm{x}_i^\top \left\{V^{-1} + \alpha I_p \right\} \bm{x}_i \right). \label{nesto}
      \end{align}
    It follows from \eqref{nesto} that $p \leq \kappa$. 
    \item[(2)] 
    From the inequality \eqref{nesto}, we have $p \leq \{1-\hat{P}_n(\bm{0})\} \kappa$, which is equivalent to $\hat{P}_n(\bm{0}) \leq 1- p / \kappa$. 
  \end{description}
\end{proof}

\subsection{Proof of Lemma~\ref{exbp}}
The main strategy for proving Lemma~\ref{exbp} is essentially the same as that adopted in Lemma 3.1 of Tyler~\cite{T8614}. 

\begin{proof}[Proof of Lemma~\ref{exbp}]
\textbf{}

  \begin{description}
      \item[(1)]
      It follows from
      \begin{align*}
        \varepsilon_m &< 1 - \frac{p}{\kappa} - \frac{p-1}{n} \\
        &< 1 - \frac{p}{\kappa} - \frac{p-1}{n+m}
      \end{align*}
      that
      \begin{align*}
        \frac{p-1+m}{n+m} &< 1 - \frac{p}{\kappa}.
      \end{align*}
      Since the good data $\mathrm{X}$ defined in Subsection~\ref{ssec:AdBP} consists of data points drawn from a continuous distribution, any $(p-1)$-dimensional linear subspace almost surely includes at most $(p-1)$ data points.
      By Corollary~\ref{cor:prexcon}, the assertion follows. 
      \item[(2)] 
      Since $\mathrm{X}$ consists of data points drawn from a continuous distribution, the assertion follows from Lemma~\ref{lem:prexnescon}-(2). 
  \end{description}
\end{proof}

\subsection{Proof of Lemma~\ref{bdbp}}
The main strategy for proving Lemma~\ref{bdbp} is essentially the same as that adopted in Theorem~1 of Tyler et al.~\cite{TYN23} and Lemma~3.2 of Tyler~\cite{T8614}. 

\begin{proof}[Proof of Lemma~\ref{bdbp}]
\textbf{}
  \begin{description}
    \item[(1)] 
    Consider $\mathrm{Z} \in \mathcal{E}_m(\mathrm{X})$ and the estimator $\hat{V}(\mathrm{Z})$. 
    Write $\mathrm{Z} = \{ \bm{z}_i ; i \in [n+m] \}$. 
    For simplicity, let $\hat{V}$ denote $\hat{V}(\mathrm{Z})$ and $\lambda_1$ the largest eigenvalue of $\hat{V}$. 
    Then
    \begin{equation*}
        \left( \frac{1}{\lambda_1} + \alpha \right) \bm{z}_i^\top \bm{z}_i \leq \bm{z}_i^\top (\hat{V}^{-1} + \alpha I_p) \bm{z}_i
    \end{equation*}
    for all $i \in [n+m]$. 
    Since Condition~\ref{conm}-(A), 
    \begin{equation*}
        w \left( \left( \frac{1}{\lambda_1} + \alpha \right) \bm{z}_i^\top \bm{z}_i \right) \geq w\left( \bm{z}_i^\top (\hat{V}^{-1} + \alpha I_p) \bm{z}_i \right) 
    \end{equation*}
    for all $i \in [n]$.
    Therefore,
    \begin{align*}
        \hat{V} &= \frac{1}{n+m} \sum_{i=1}^{n+m} w \left( \bm{z}_i^\top (\hat{V}^{-1} + \alpha I_p) \bm{z}_i \right) \bm{z}_i \bm{z}_i^\top \\
        &\preceq \frac{1}{n+m} \sum_{i=1}^{n+m} w \left( \left( \frac{1}{\lambda_1} + \alpha \right) \bm{z}_i^\top \bm{z}_i \right) \bm{z}_i \bm{z}_i^\top \\
        &= \frac{1}{(n+m) \left( 1 / \lambda_1 + \alpha \right)} \sum_{i=1}^{n+m} \psi \left( \left( \frac{1}{\lambda_1} + \alpha \right) \bm{z}_i^\top \bm{z}_i \right) \frac{\bm{z}_i \bm{z}_i^\top}{\bm{z}_i^\top \bm{z}_i} \\
        &\preceq \frac{\kappa}{(n+m) \left( 1 / \lambda_1 + \alpha \right)} \sum_{i=1}^{n+m} \frac{\bm{z}_i \bm{z}_i^\top}{\bm{z}_i^\top \bm{z}_i} \\
        &\preceq \frac{\kappa}{ 1 / \lambda_1 + \alpha} I_p,
    \end{align*}
    which implies that
    \begin{equation*}
        \lambda_1 \leq \frac{\kappa}{ 1 / \lambda_1 + \alpha}.
    \end{equation*}
    Lemma~\ref{lem:prexnescon}-(1) implies $\kappa \geq 1$, and hence $\lambda_1 \leq (\kappa - 1) / \alpha$. 

    \item[(2)]
    Consider the sequence $\{ \mathrm{Y}_\ell \}_{\ell \in \mathbb{N}}$ that satisfies $\mathrm{Z}_\ell = \mathrm{X} \cup \mathrm{Y}_\ell \in \mathcal{E}_m(\mathrm{X})$ for $\ell \in \mathbb{N}$. 
    Let
    \begin{align*}
      \omega_\ell &\coloneqq \mathrm{Tr}[V^{-1}(\mathrm{Z}_\ell)] \, (\neq 0), \\
      \Omega_\ell &\coloneqq \frac{V^{-1}(\mathrm{Z}_\ell)}{\omega_\ell}.
    \end{align*}
    By the compactness of $\mathcal{S}_{+}^{p*} \coloneqq \{ X \in \mathcal{S}_{+}^p ; \, \mathrm{Tr}[X] = 1\}$ in the Euclidean topology, the sequence $\{\Omega_\ell\}_{\ell \in \mathbb{N}}$ admits a convergent subsequence $\{\Omega_j\}_{j \in J}$.
    Hereafter, we shall denote this limit operation by
    \begin{equation*}
      \lim_{j \to \infty} \Omega_j = \Omega \in \mathcal{S}_{+}^{p*}.
    \end{equation*} 
    
    Let $\bm{e}$ be a vector satisfying that $\| \bm{e} \| = 1$ and $\Omega \bm{e} \neq \bm{0}$, where such $\bm{e}$ exists because of $\Omega \neq 0_{p \times p}$. 
    Write $\mathrm{Y}_\ell = \{ \bm{y}_{i, \ell} \in \mathbb{R}^p ; i \in [m] \}$, and set
    \begin{align*}
      r_{i, \ell} &\coloneqq \left( \bm{y}_{i,\ell}^\top \bm{y}_{i,\ell} \right)^{1/2} , \\
      \bm{u}_{i,\ell} &\coloneqq 
      \begin{cases}
        \frac{\bm{y}_{i,\ell}}{r_{i,\ell}} & \text{($r_{i,\ell} \neq 0$)}, \\
        \bm{e} & \text{(otherwise)}, 
      \end{cases}
    \end{align*}
    for $i \in [m]$. 
    Since $\mathbb{S}^{p-1}$ is compact in the Euclidean topology, the product space $\mathcal{S}_{+}^{p*} \times \mathbb{S}^{p-1}$ is also compact in the product topology, and hence there exists an index set $J$ and a point $\bm{u}_i \in \mathbb{S}^{p-1}$ such that 
    \[  
        \lim_{j \to \infty} \bm{u}_{i,j} = \bm{u}_i. 
    \]
    Let $\bar{P}_{\Omega} \coloneqq I_p - \Omega \Omega^\dagger $ denote the orthogonal projection onto $\textnormal{Ker}(\Omega)$. 
    It holds that
    \begin{align*}
      & V_\ell = \frac{1}{n+m} \sum_{i=1}^{n+m} w(\bm{z}_i^\top (V_\ell^{-1} + \alpha I_p) \bm{z}_i) \bm{z}_i \bm{z}_i^\top \\
      \iff& I_p = \frac{1}{n+m} \sum_{i=1}^{n+m} w(\bm{z}_i^\top (\omega_\ell \Omega_\ell + \alpha I_p) \bm{z}_i) \omega_\ell \Omega_\ell^{1/2} \bm{z}_i \bm{z}_i^\top \Omega_\ell^{1/2} \\
      \iff& \bar{P}_\Omega = \frac{1}{n+m} \sum_{i=1}^{n+m} w(\bm{z}_i^\top (\omega_\ell \Omega_\ell + \alpha I_p) \bm{z}_i) \omega_\ell \bar{P}_\Omega \Omega_\ell^{1/2} \bm{z}_i \bm{z}_i^\top \Omega_\ell^{1/2} \bar{P}_\Omega, 
    \end{align*}
    which implies
    \begin{equation}
      \textnormal{rank}(\bar{P}_\Omega) = \frac{1}{n+m} \sum_{i=1}^{n+m} w(\bm{z}_i^\top (\omega_\ell \Omega_\ell + \alpha I_p) \bm{z}_i) \omega_\ell \bm{z}_i^\top \Omega_\ell^{1/2} \bar{P}_\Omega \Omega_\ell^{1/2} \bm{z}_i \label{lbdtc}. 
    \end{equation}
    Hence, \eqref{lbdtc} is a necessary condition for $\mathrm{Z} \in \mathcal{E}_m(\mathrm{X})$. 
    If $\Omega \bm{x}_i \neq \bm{0}$, then
    \begin{align*}
      & w(\bm{x}_i^\top (\omega_j \Omega_j + \alpha I_p) \bm{x}_i) \omega_j \bm{x}_i^\top \Omega_j^{1/2} \bar{P}_\Omega \Omega_j^{1/2} \bm{x}_i \\
      &= \psi(\bm{x}_i^\top (\omega_j \Omega_j + \alpha I_p) \bm{x}_i) \frac{\bm{x}_i^\top \Omega_j^{1/2} \bar{P}_\Omega \Omega_j^{1/2} \bm{x}_i}{\bm{x}_i^\top \left\{\Omega_j + \left(\alpha / \omega_j \right) I_p\right\} \bm{x}_i}. 
    \end{align*}
    Recalling Condition~\ref{conm}-(B), it follows from $\bar{P}_{\Omega} \Omega^{1/2} = 0_{p \times p}$ that
    \begin{equation*}
      \lim_{j \to \infty} \psi(\bm{x}_i^\top (\omega_j \Omega_j + \alpha I_p) \bm{x}_i) \frac{\bm{x}_i^\top \Omega_j^{1/2} \bar{P}_\Omega \Omega_j^{1/2} \bm{x}_i}{\bm{x}_i^\top \left\{\Omega_j + \left(\alpha / \omega_j \right) I_p\right\} \bm{x}_i} = 0. 
    \end{equation*}
    In a similar manner, if $\Omega \bm{u}_i \neq \bm{0}$, then
    \begin{equation*}
      \lim_{j \to \infty} \psi(\bm{y}_{i,j}^\top (\omega_j \Omega_j + \alpha I_p) \bm{y}_{i,j}) \frac{\bm{u}_{i,j}^\top \Omega_j^{1/2} \bar{P}_\Omega \Omega_j^{1/2} \bm{u}_{i,j}}{\bm{u}_{i,j}^\top \left\{\Omega_j + \left(\alpha / \omega_j \right) I_p\right\} \bm{u}_{i,j}} = 0.
    \end{equation*}
    Combining these and \eqref{lbdtc}, which is a necessary condition for $\mathrm{Z} \in \mathcal{E}_m(\mathrm{X})$, gives
    \begin{equation}\label{lbdtc3}
      \textnormal{rank}(\bar{P}_\Omega) = \lim_{j \to \infty} \frac{1}{n+m} \sum_{\bm{z} \in \mathrm{Z}_{0,j}} \psi(\bm{z}^\top (\omega_j \Omega_j + \alpha I_p) \bm{z}) \frac{\bm{z}^\top \Omega_j^{1/2} \bar{P}_\Omega \Omega_j^{1/2} \bm{z}}{\bm{z}^\top \left\{\Omega_j + \left(\alpha / \omega_j \right) I_p\right\} \bm{z}},
    \end{equation}
    where $\mathrm{Z}_{0,j} = \{ \bm{x}_i \in \mathrm{X} ; \, \Omega \bm{x}_i = \bm{0} \} \cup \{ \bm{y}_{i,j} \in \mathrm{Y}_j ; \, \Omega \bm{u}_i = \bm{0} \}$. 
    From \eqref{nesto0}, which is another necessary condition for $\mathrm{Z} \in \mathcal{E}_m(\mathrm{X})$, it follows that
    \begin{align*}
      p &\geq \lim_{j \to \infty} \frac{1}{n+m} \left\{ \sum_{\bm{z} \in \mathrm{Z}_{0,j}} \psi\left(\bm{z}^\top (\omega_j \Omega_j + \alpha I_p) \bm{z}\right) \frac{\bm{z}^\top \Omega_j \bm{z}}{\bm{z}^\top \left\{\Omega_j + \left(\alpha / \omega_j \right) I_p\right\} \bm{z}} \right.\\
      &\hspace{7em} \left. + \sum_{\bm{x}_i \in \mathrm{X} \cap \textnormal{Ker}(\Omega)^c} \psi\left(\bm{x}_i^\top (\omega_j \Omega_j + \alpha I_p) \bm{x}_i\right) \frac{\bm{x}_i^\top \Omega_j \bm{x}_i}{\bm{x}_i^\top \left\{\Omega_j + \left(\alpha / \omega_j \right) I_p\right\} \bm{x}_i} \right\} \\
      &\geq \lim_{j \to \infty} \frac{1}{n+m} \left\{ \sum_{\bm{z} \in \mathrm{Z}_{0,j}} \psi\left(\bm{z}^\top (\omega_j \Omega_j + \alpha I_p) \bm{z}\right) \frac{\bm{z}^\top \Omega_j^{1/2} \bar{P}_\Omega \Omega_j^{1/2} \bm{z}}{\bm{z}^\top \left\{\Omega_j + \left(\alpha / \omega_j \right) I_p\right\} \bm{z}} \right. \\
      &\hspace{7em} + \left. \sum_{\bm{x}_i \in \mathrm{X} \cap \textnormal{Ker}(\Omega)^c} \psi\left(\bm{x}_i^\top (\omega_j \Omega_j + \alpha I_p) \bm{x}_i\right) \frac{\bm{x}_i^\top \Omega_j \bm{x}_i}{\bm{x}_i^\top \left\{\Omega_j + \left(\alpha / \omega_j \right) I_p\right\} \bm{x}_i} \right\}.
    \end{align*}
    From \eqref{lbdtc3}, we obtain
    \begin{align*}
      p &\geq \textnormal{rank}(\bar{P}_\Omega) + \frac{1}{n + m} \sum_{\bm{x}_i \in \mathrm{X} \cap \textnormal{Ker}(\Omega)^c} \lim_{j \to \infty} \psi\left(\bm{x}_i^\top (\omega_j \Omega_j + \alpha I_p) \bm{x}_i\right) \frac{\bm{x}_i^\top \Omega_j \bm{x}_i}{\bm{x}_i^\top \left\{\Omega_j + \left(\alpha / \omega_j \right) I_p\right\} \bm{x}_i}.
    \end{align*}
    Letting $\omega_j \to + \infty \, (j \to \infty)$ yields
    \begin{align*}
      p &\geq \textnormal{rank}(\bar{P}_\Omega) + \frac{1}{n + m} \sum_{\bm{x}_i \in \mathrm{X} \cap \textnormal{Ker}(\Omega)^c} \kappa. 
    \end{align*}
    From $\# \{ \bm{x} \in \mathrm{X} \cap \textnormal{Ker}(\Omega) \} \leq \textnormal{rank}(\bar{P}_\Omega)$ almost surely, it follows that
    \begin{align}
      & p \geq \textnormal{rank}(\bar{P}_\Omega) + \frac{n - \textnormal{rank}(\bar{P}_\Omega)}{n + m} \kappa \notag\\
      \iff& \varepsilon_m \geq 1 - \frac{n\{p - \textnormal{rank}(\bar{P}_\Omega)\}}{\kappa \{n - \textnormal{rank}(\bar{P}_\Omega)\}}. \label{lbdtc4}
    \end{align}
    Since the right-hand side of \eqref{lbdtc4} is increasing with respect to $\textnormal{rank}(\bar{P}_\Omega)$, it holds that $\varepsilon_m \geq 1 - \kappa / p$ for $\textnormal{rank}(\bar{P}_\Omega) = 0$.
    The conclusion follows by the contrapositive of the above. 
  \end{description}
\end{proof}

\subsection{Proof of Theorem~\ref{bpbp}}
The main strategy for proving Theorem~\ref{bpbp} is essentially the same as that adopted in Theorem~3.1 of Tyler~\cite{T8614}. 

\begin{proof}[Proof of Theorem~\ref{bpbp}]
When $\varepsilon_m > 1 - p / \kappa$, setting $Y = \{ \bm{0}, \ldots, \bm{0} \}$ yields $\hat{P}_{n+m}(\{\bm{0}\}) > 1 - p / \kappa$, which in turn implies that $Z \notin \mathcal{E}_m(\mathrm{X})$, and hence $\hat{V}(\mathrm{Z})$ breaks down, in accordance with Lemma \ref{exbp}-(2).
When $\varepsilon_m < 1 - p / \kappa - (p-1) / n$, from Lemmas~\ref{exbp}-(1) and \ref{bdbp}, we see that $\hat{V}(\mathrm{Z})$ does not breakdown. 
\end{proof}

\section{Concluding remarks}\label{summary}
\begin{table}[t]
    \centering
    \includegraphics[width=\textwidth]{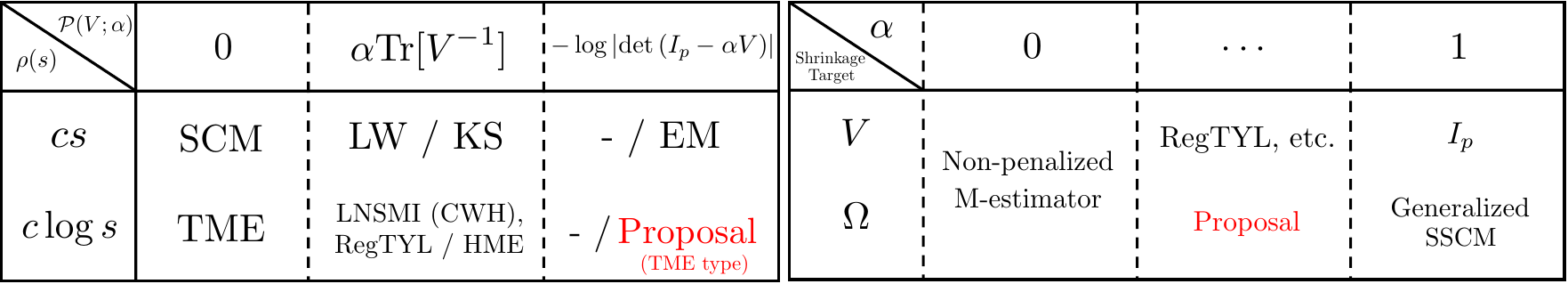}
    \caption{Correspondence table of estimators, where $c$ denotes an arbitrary constant. }
    \label{tab:summary}
\end{table}

We proposed an estimator of the scatter matrix for elliptical distributions. 
The proposed estimator is obtained by incorporating an EM-type estimator into the precision structure estimation step of the classical $M$-estimation framework. 
We investigated conditions for the existence of a solution to the estimating equation~\eqref{eseq:pro1} in $\mathcal{S}_{++}^p$, and discussed robustness of the proposed estimator through the concept of AdBP. 
In particular, Lemma~\ref{bdbp}-(1) shows that the proposed estimator does not \textit{breakdown from above}, which is a distinctive robustness of the proposed estimator. 

We further examined the statistical properties of the proposed method from the perspectives of penalized $M$-estimation and Bayesian estimation. 
Moreover, as summarized in Table~\ref{tab:summary}, we discussed correspondences between the proposed estimator and various existing estimators. 
Through numerical experiments, we demonstrated that the performance of TME and LNSMI deteriorates in the presence of clustered outliers, whereas the proposed Tyler type estimator remains stable and provides accurate estimation even in such settings. 
Based on arguments similar to those in the existing literature, we also proposed a method for selecting the shrinkage parameter of the proposed Tyler type estimator. 

Finally, we outlined several directions for future work. 
These include, but are not limited to, the uniqueness of solutions to the minimization problem~\eqref{eq:pmf20} and the global convergence properties of the proposed estimator. 

\section*{Acknowledgements}
The second author was supported by Japan Society for the Promotion of Science KAKENHI Grant Numbers JP23K16851 and the Research Fellowship Promoting International Collaboration of the Mathematical Society of Japan.
The third author was supported by Japan Society for the Promotion of Science KAKENHI Grant Number 25K07133.

\startappendix

\section{Reason for focusing on precision estimation in TME and LNSMI}
From Figure~\ref{fig:r1}-(b, d, f), we observe that the estimation accuracy of TME and LNSMI decreases monotonically but in a nonlinear manner, whereas that of SCM decreases monotonically and linearly.
Based on this observation, we suspected that the difficulty lies in the part of the procedure related to estimating the precision structure.

Indeed, the following observations were obtained, which suggest that difficulties arise specifically in the estimation of the precision structure. 

First, we generate data containing outliers using the procedure described in Section~\ref{sec:simusets}, and record the indices corresponding to the outlier points. 
Next, based on the generated data, we compute TME and use the resulting estimate $\hat{V}_{\textnormal{TME}}$ to calculate the Mahalanobis distance $\sqrt{\bm{x}_i^\top \hat{V}_\textnormal{TME}^{-1} \bm{x}_i} \; (i \in [n])$ for each data point. 
We then examine the indices of the data points with the largest Mahalanobis distances under this estimate. 

As a result, we found that once the dimension exceeds the point at which the estimator’s performance begins to deteriorate sharply, the indices of data points not designated as outliers start to appear above the designated outlier indices. 
This indicates that, toward the end of the algorithm, data points that are not true outliers are being treated as outliers, while those originally designated as outliers are being treated as non-outliers. 

\section{Number of resampling iterations}
This type of discussion is often referred to as the \textit{coupon collector problem}. 
For sampling with replacement from a dataset of size $N$, the expected number of resampling iterations $m_k$ required to obtain $k$ distinct samples ($k \in \mathbb{N}$, $k \leq N$) is 
\begin{equation*}
    \mathrm{E}[m_k] = N \left( \sum_{i = N-k+1}^N \frac{1}{i} \right) \text{.}
\end{equation*}
For example, consider the case where $k = \lceil 0.975 N \rceil$. 

\section[Generation of the true covariance matrix]{Generation of the true covariance matrix $\Sigma(c_1, c_2)$}

\begin{algorithm}[H]
    \caption{Generation method for the 1st through $p$-th contribution ratios}
    \begin{algorithmic}
        \STATE \textbf{Input:} $p \ge 4$, $c_1$, $c_2$, $s=1$
        \STATE \textbf{begin}
            \STATE \hspace{1em} {if $\left(\left( (p-1)c_1 + \frac{c_1}{c_2} < 1 \right) \textnormal{or} \left(c_1+(p-1)\frac{c_1}{c_2} > 1 \right)\right)$ then}
                \STATE \hspace{2em} {$\bm{r} = \left(\frac{1}{p}, \ldots, \frac{1}{p}\right) \in \Delta^{p-1}$}
            \STATE \hspace{1em} {else then}
                \STATE \hspace{2em} {while ($s \leq 10^5$) do}
                    \STATE \hspace{3em} {$\bm{l} \coloneqq (l_1, \ldots, l_{p-2})^{\top} \sim \textnormal{Dir}_{p-2}(1)$}
                    \STATE \hspace{3em} {if $\left(\frac{c_1/c_2}{1-c_1/c_2-c_1} \leq \min\left\{l_1, \ldots, l_{p-2}\right\} \right)$ and $\left(\max\left\{l_1, \ldots, l_{p-2}\right\} \leq \frac{c_1}{1-c_1/c_2-c_1}\right)$ then break;}
                    \STATE \hspace{3em} {$s \leftarrow s+1$}
                \STATE \hspace{2em} {end}
                \STATE \hspace{2em} {if $s = 100001$ then}
                    \STATE \hspace{3em} {$\bm{l} = \left(\frac{1}{p-2}, \ldots, \frac{1}{p-2}\right) \in \Delta^{p-3}$}
                \STATE \hspace{2em} {end}
                \STATE \hspace{2em} {$\bm{r} = \textnormal{sort}\left( c_1, \frac{c_1}{c_2}, \bm{l}^{\top} \right) \in \Delta^{p-1}$}
            \STATE \hspace{1em} {end} 
        \STATE \textbf{end}
        \STATE \textbf{Return:} {$\bm{r} \in \Delta^{p-1}$}
    \end{algorithmic}
\end{algorithm}

In the numerical experiments, the true covariance matrix $\Sigma(c_1, c_2)$ is generated randomly under specified values of the largest contribution ratio $c_1$, the condition number $c_2$, and $\mathrm{Tr}[\Sigma(c_1, c_2)]$. 
Once the $p$ eigenvectors of $\Sigma(c_1, c_2)$ and the contribution ratios from the second to the $(p-1)$-th components are appropriately generated, the true covariance matrix $\Sigma(c_1, c_2)$ can be randomly constructed. 

The $p$ eigenvectors of the true covariance matrix are obtained by generating $p$ random samples from $\mathcal{N}_p(\bm{0}, I_p)$ and orthogonalizing them using the Gram--Schmidt procedure.
The contribution ratios, on the other hand, are generated as described in Algorithm~1. 

In Algorithm~1, we first branch according to whether a matrix satisfying the specified largest contribution ratio $c_1$ and smallest contribution ratio $c_1 / c_2$ can exist.
If such a matrix exists, the contribution ratios are generated from $\textnormal{Dir}_{p-2}(1)$ in principle. 
However, depending on the values of $c_1$ and $c_2$, generating random samples from $\textnormal{Dir}_{p-2}(1)$ may become difficult. 
In such cases, all contribution ratios from the second to the $(p-1)$-th components are set equal. 

\newpage
\section{Table of RMSE values corresponding to Figure~\ref{fig:r1}}
\begin{table}[H]
    \caption{Table of RMSE values corresponding to Figure~\ref{fig:r1}. }
    \centering
    \includegraphics[height=0.9\textheight]{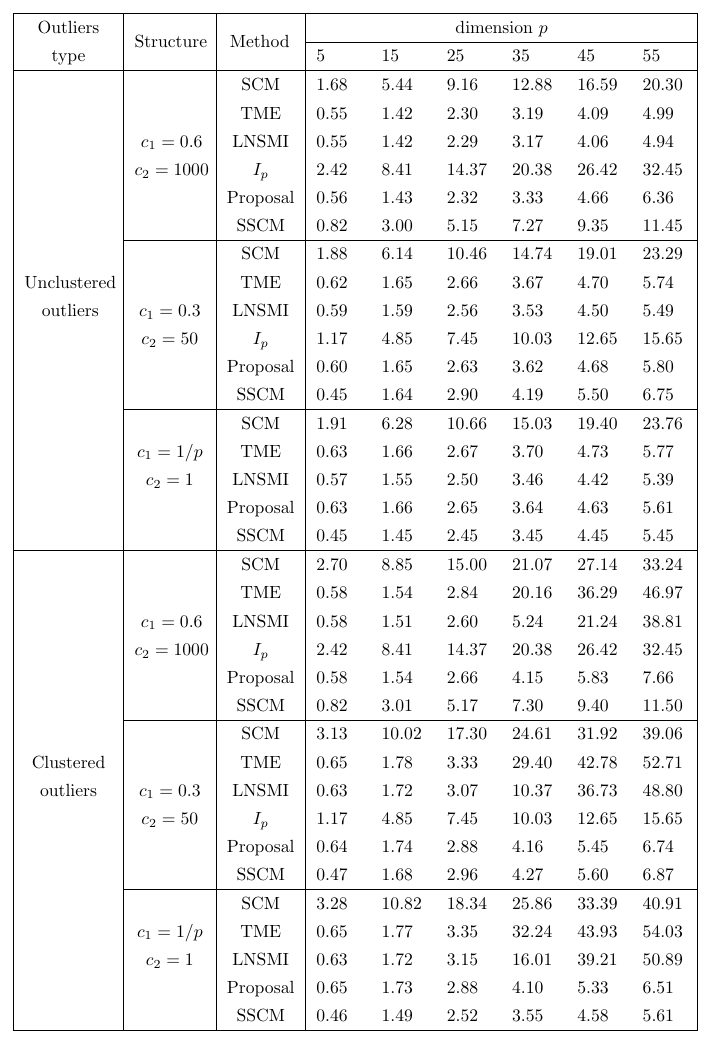}\label{tab:tab}
\end{table}

Table~\ref{tab:tab} reports the RMSE values corresponding to Figure~\ref{fig:r1}. 

\section{Numerical results under various scenarios -1}
Under the numerical settings described in Sections~\ref{sec:simusets} and~\ref{coef}, we additionally present numerical results for several alternative parameter configurations, following the same format as in Section~\ref{sec:simr1}. 

\setcounter{secnumdepth}{3} 
\renewcommand{\thesubsubsection}{\arabic{subsubsection}} 
\subsection[Comparison across outlier ratios]{Comparison across $\xi$}

\begin{figure}[h]
    \centering
    \includegraphics[width=0.90\columnwidth]{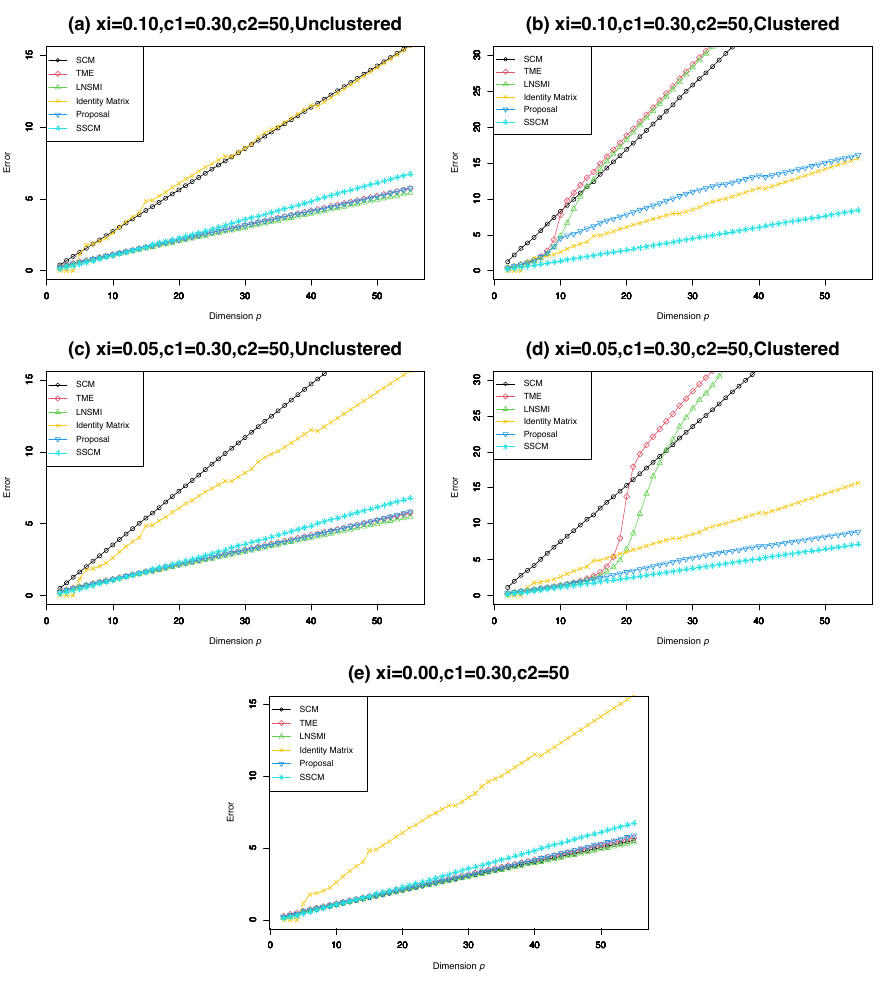}
    \caption{
    Plots of the RMSE values of each estimator for different data dimensions $p$, with lines connecting the plotted points. 
    The horizontal axis represents the dimension $p$ $( = 2, \ldots, 55 )$, and the vertical axis represents the RMSE. 
    Figure~(a, c) show the results for unclustered outliers, Figure~(b, d) correspond to clustered outliers, and Figure~(e) represents the case with no outliers, that is, the situation in which all data points are simply generated as random samples from a multivariate normal distribution. 
    The other settings are $N = 100$ and $k = 10$. In Figure~(a, c, e), the curves representing TME, LNSMI, and the proposed Tyler type estimator nearly overlap. 
    Furthermore, in Figure~(a), the curves representing SCM and $I_p$, in Figure~(b), those representing TME and LNSMI, and in Figure~(e), those representing SCM, TME, LNSMI, and the proposed Tyler type estimator also nearly overlap. 
    }\label{fig:r4}
\end{figure}

We now discuss the numerical results. 
First, as shown in Figure~\ref{fig:r4}, a substantial degradation in the performance of the TME and the LNSMI is again observed when clustered outliers are present. 

Figure~\ref{fig:r4}-(b, d, e) illustrates how the estimator performance varies with respect to the proportion of clustered outliers $\xi$, revealing a characteristic transition behavior. 
Specifically, in Figure~\ref{fig:r4}-(d), when $\xi = 0.03$, a pronounced deterioration in the performance of the TME occurs around dimension $p = 33$, whereas in Figure~\ref{fig:r4}-(b), when $\xi = 0.1$, a similar phenomenon appears already around $p = 10$. 
These numerical observations are consistent with the fact that the asymptotic breakdown point of the TME lies between $1/(p+1)$ and $1/p$ \cite{DT05}. 

\subsection[Comparison across magnitude of outliers]{Comparison across $k$}
\begin{figure}[h]
    \centering
    \includegraphics[width=0.90\columnwidth]{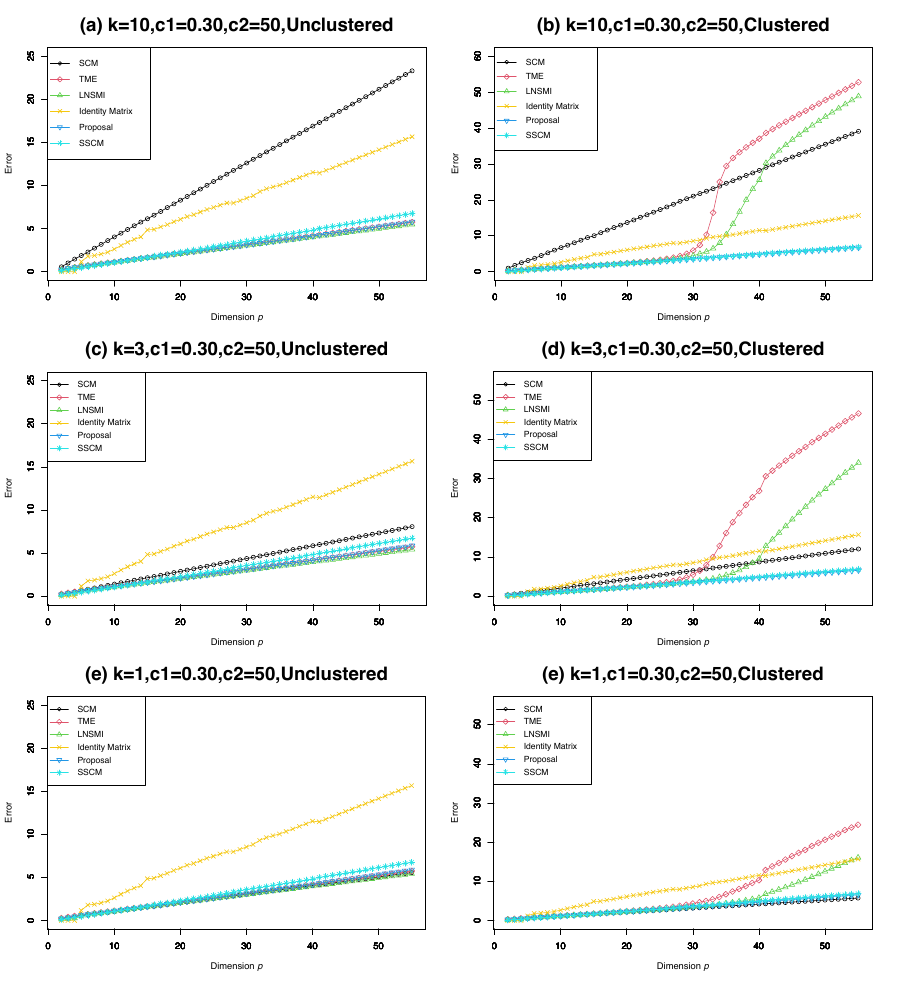}
    \caption{
    Plots of the RMSE values of each estimator for different data dimensions $p$, with lines connecting the plotted points. The horizontal axis represents the dimension $p$ $( = 2, \ldots, 55 )$, and the vertical axis represents the RMSE. 
    Figure~(a, c, e) show the results for unclustered outliers, while Figure~(b, d, f) correspond to clustered outliers. 
    The other settings are $N = 100$ and $\xi = 0.03$. 
    In Figure~(a, c, e), the curves representing TME, LNSMI, and the proposed Tyler type estimator nearly overlap. 
    Furthermore, in Figure~(b, d, f), the curves representing SSCM and the proposed Tyler type estimator also nearly overlap. 
    }\label{fig:r5}
\end{figure}

First, as shown in Figure~\ref{fig:r5}, a substantial degradation in the performance of the TME and the LNSMI is again observed when clustered outliers are present. 

Next, Figure~\ref{fig:r5}-(b, d, f) presents the estimator performance as a function of the parameter $k$, which controls the severity of clustered outliers. 
A similar transition behavior can also be observed in this case. 
Notably, even when $k = 1$, that is, in situations where no data points have exceptionally large $\ell_2$-norms, a substantial degradation in the performance of both the TME and the LNSMI is still observed.

Finally, although the behavior of the estimators with respect to variations in the maximum contribution ratio $c_1$, the condition number $c_2$, and the sample size $N$ was also examined, the corresponding results are omitted since no particularly notable changes were observed. 
In these settings as well, the presence of clustered outliers consistently leads to a marked deterioration in the performance of the TME and the LNSMI. 
As the true covariance structure deviates further from the identity matrix, the performance of the SSCM tends to deteriorate, but no additional noteworthy insights were obtained beyond these observations. 

\section{Numerical results under various scenarios -2}
Following existing studies \cite{CWH11, M76, T87, W12a}, we also report numerical results obtained under the $t$-distribution. 
As in the existing literature, these results indicate that the proposed Tyler type estimator successfully attains robustness even when the data follow a $t$-distribution. 

In this section, the data, including outliers, are generated from a $t$-distribution. 
This setting corresponds to the case of unclustered outliers, and none of the estimators exhibit breakdown behavior. 
The curves representing the performance of TME, LNSMI, SSCM, and the proposed Tyler type estimator nearly overlap. 

\begin{figure}[h]
    \centering
    \includegraphics[width=0.95\columnwidth]{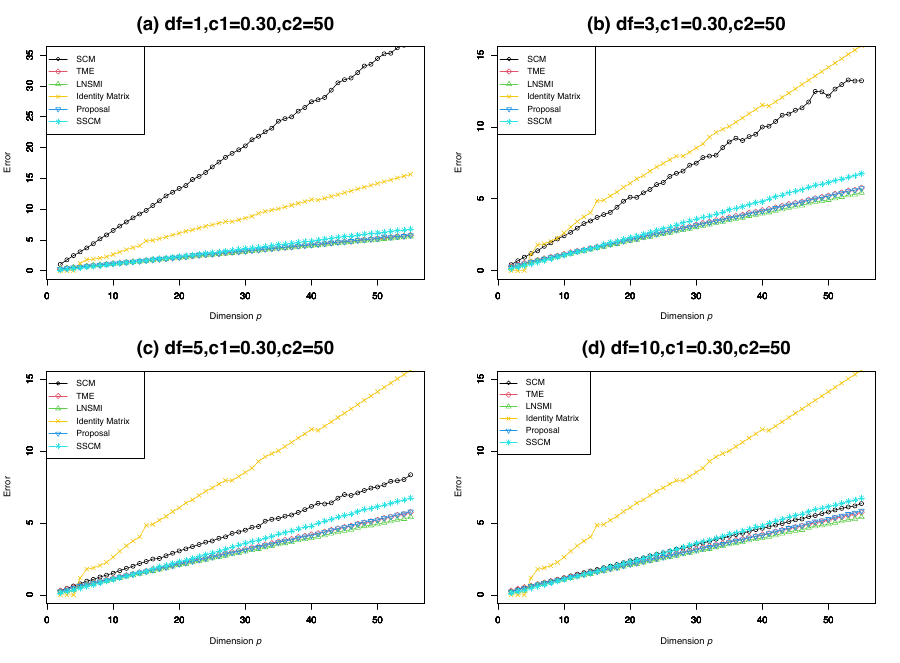}
    \caption{
    Plots of the RMSE values of each estimator for different data dimensions $p$, with lines connecting the plotted points. 
    The horizontal axis represents the dimension $p$ $( = 2, \ldots, 55 )$, and the vertical axis represents the RMSE. Figure~(a) shows the results when the data are generated solely from a $t$ distribution with degree of freedom $1$, Figure~(b) from a $t$ distribution with degree of freedom 3, Figure~(c) from a $t$ distribution with degree of freedom 5, and Figure~(d) from a $t$ distribution with degree of freedom 10. 
    Note that no outliers are generated separately in this setting. In Figure~(a), the curves representing TME, LNSMI, and the proposed Tyler type estimator nearly overlap. 
    In Figure~(b, c, d), the curves representing TME and the proposed Tyler type estimator nearly overlap. 
    }\label{fig:r0}
\end{figure}

\end{document}